\documentclass{article}
\usepackage{graphicx} 
\usepackage{amssymb}
\usepackage{amsfonts}
\usepackage{graphicx}
\usepackage{amscd}
\usepackage{amsmath}
\usepackage{pictex}
\usepackage{amsthm}
\usepackage{authblk}
\newtheorem{theorem}{Theorem}[section]
\newtheorem{lemma}[theorem]{Lemma}
\theoremstyle{definition}

\newtheorem{example}[theorem]{Example}
\newtheorem{proposition}[theorem]{Proposition}
\theoremstyle{remark}
\newtheorem{remark}[theorem]{Remark}
\numberwithin{equation}{section}

\begin{document}

\title{Koba-Nielsen local zeta functions, convex subsets, and  generalized Selberg-Mehta-Macdonald and Dotsenko-Fateev-like integrals}


\author {Willem Veys\footnote{KU Leuven, Department of Mathematics, Celestijnenlaan 200 B, B-3001 Leuven,

Belgium. E-mail: wim.veys@kuleuven.be.

The author was supported by KU Leuven Grant GYN-E4282-C16/23/010.} \ and W. A. Z\'{u}\~{n}iga-Galindo\footnote{University of Texas Rio Grande Valley, School of Mathematical and Statistical Sciences, 

One West University Blvd,  Brownsville, TX 78520, United States. 

E-mail: wilson.zunigagalindo@utrgv.edu.

The author was partially funded by L. Debnath Endowment}}
\date{}
\maketitle








\begin{abstract}
    The Koba-Nielsen local zeta functions are integrals depending on several complex parameters, used to regularize the Koba-Nielsen string amplitudes. These integrals are convergent and admit meromorphic continuations in the complex parameters.  In the original case, the integration is carried out on the $n$-dimensional Euclidean space. In this work, the integration is over a variety of (bounded or unbounded) convex subsets; the resulting integrals also admit meromorphic continuations in the complex parameters. We describe the meromorphic continuation's polar locus explicitly, using the technique of embedded resolution.  This result can be reinterpreted as saying that the meromorphic continuations are weighted sums of Gamma functions, evaluated at linear combinations of the complex parameters, where the weights are holomorphic functions. The integrals announced in the title of this paper occur as a particular case of these new Koba-Nielsen local zeta functions, or of a further generalization to arbitrary hyperplane arrangements.
\end{abstract}


\section{Introduction} 

The Selberg-Mehta-Macdonald and Dotsenko-Fateev-like integrals play a central
role in several areas in mathematics and physics, for instance, in random
matrix theory, multivariable orthogonal polynomial theory,
Calogero--Sutherland quantum many-body systems, Knizhnik--Zamolodchikov
equations, among other areas; see e.g. \cite{Forrester et al}-\cite{Sutherland 2}.

In this paper, we give a unified approach to a large class of these
integrals, showing that they are particular cases of {\em  generalized
Koba-Nielsen local zeta functions} $Z_{\varphi }^{(N)}\left( D;\boldsymbol{s}
\right) $.\ These integrals are defined as 
\begin{equation}
Z_{\varphi }^{(N)}\left( D;\boldsymbol{s}\right) :={\int\limits_{D}}\varphi
(x){\prod\limits_{i=1}^{N}}\left\vert x_{i}\right\vert ^{s_{0i}}{
\prod\limits_{i=1}^{N}}\left\vert 1-x_{i}\right\vert ^{s_{i(N+1)}}\text{ }{
\prod\limits_{1\leq i<j\leq N}}\left\vert x_{i}-x_{j}\right\vert ^{s_{ij}}dx,
\label{zeta_function_string}
\end{equation}
where $N\geq 1$, 
$s_{0i}$ and $s_{i(N+1)}$ for $1\leq i\leq N$ and $s_{ij}$ for $1\leq i<j\leq N$ are complex variables, that is, 
$\boldsymbol{s}:=\left( s_{ij}\right) \in \mathbb{C}^{\boldsymbol{d}}$ with $\boldsymbol{d}=\frac{N(N+3)}{2}$, 
and the function $\varphi :{\mathbb{R}}^{N}\rightarrow {
\mathbb{C}}$ is smooth on the (closed) integration domain $D$. The
integration domain $D$ is \emph{any} polyhedron with boundary conditions
given by inequalities of the form $x_{i}\geq 0$, $x_{i}\leq 0$, $x_{i}\geq 1 
$ or $x_{i}\leq 1$ for some indices $i$ and/or $x_{i}\geq x_{j}$ for some $
i\neq j$, such that $\dim (D)=N$. Here we include the case `no conditions',
being $D={\mathbb{R}}^{N}$.

\smallskip
In the case $D={\mathbb{R}}^{N}$ and $\varphi (x)\equiv 1$, the Koba-Nielsen
local zeta function $Z_{\varphi }^{(N)}\left( D;\boldsymbol{s}\right) $ was
studied in \cite{BVZ-JHP}, see also \cite{BGZ-LMP}-\cite{Symmetry}. In that work, it was established that these
integrals are convergent and  holomorphic in a nonempty open subset of $\mathbb{C}^{\boldsymbol{d}}$. Furthermore, they admit meromorphic
continuations to $\mathbb{C}^{\boldsymbol{d}}$, having a polar locus
consisting of a finite union of hyperplanes in $\mathbb{C}^{\boldsymbol{d}}$,
see \cite[Theorem 4.1, Proposition 5.2]{BVZ-JHP}.  These integrals were used as regularizations of Koba-Nielsen string amplitudes. In  \cite{BVZ-JHP}, their study was based on Hironaka's desingularization theorem \cite{H}, and techniques of multivariate local zeta functions, see, e.g., \cite{Igusa-old}-\cite{Loeser}.

In Section \ref{resolution}, we review Hironaka's desingularization
theorem. It is roughly some kind of change of variables procedure,
transforming the integrand via a map $\pi $ into a function that is
essentially monomial (in local coordinates). Globally, the total inverse
image of 
\begin{equation}\label{A_N}
A_{N}(x)=\left\{ x\in \mathbb{R}^{N}\mid {\prod\limits_{i=1}^{N}}x_{i}
\text{ }{\prod\limits_{i=1}^{N}}\left( 1-x_{i}\right) \text{ }{
\prod\limits_{1\leq i<j\leq N}}\left( x_{i}-x_{j}\right) =0\right\} 
\end{equation}
by $\pi $ is a union of nonsingular $(N-1)$-dimensional manifolds $
E_{i},i\in T,$ that intersect each other transversally.

In \cite{BVZ-JHP}, we showed that, when $D={\mathbb{R}}^{N}$ and $\varphi
(x)\equiv 1$, the polar locus of $Z_{\varphi }^{(N)}\left( D;\boldsymbol{s}
\right) $ is included in a set of hyperplanes in ${\mathbb{C}}^{\boldsymbol{d
}}$, induced by these $E_{i}$, and we determined precisely these possible polar
hyperplanes.  In fact, thanks to some kind of `universality'\ of the
embedded resolution construction, this implies that, \emph{for any polyhedron $D$ as
above}, the polar locus of $Z_{\varphi }^{(N)}\left( D;\boldsymbol{s}\right) 
$ is contained in this same list of hyperplanes. The main result of the
present paper is to determine which $E_{i}$ effectively contribute, see Theorem \ref{maintheorem} and Proposition \ref{independence}.

\smallskip
Our general theory covers for instance the setting of Sussman \cite{Sussman}, where the existence of meromorphic continuations
for several types of Dotsenko-Fateev-like integrals is established. These integrals have
the form $Z_{\varphi }^{(N)}\left( D;\boldsymbol{s}\right) $, where $D$ is
the standard $N$-simplex 
\begin{equation*}
\Delta _{N}=\{(x_{1},\dots ,x_{n})\in {\mathbb{R}}^{N}\mid 0\leq x_{1}\leq
x_{2}\leq \dots \leq x_{N}\leq 1\},
\end{equation*}
or 
\begin{equation*}
\square _{N}=\{(x_{1},\dots ,x_{n})\in {\mathbb{R}}^{N}\mid 0\leq x_{i}\leq 1
\text{ for }i=1,\dots ,N\}.
\end{equation*}
Our main results allow us to recover in a conceptual way Sussman's results, see Section \ref{Sussman-Section}.

Many of the classical results about Selberg, Mehta, and Macdonald integrals
provide explicit formulas for meromorphic continuations for some integrals of type $Z_{\varphi }^{(N)}\left( D;\boldsymbol{s}\right) $, in terms of Gamma
functions evaluated at linear combinations of the variables $s_{ij}$. For 
general smooth functions $\varphi $, and \ domains like the ones considered here, such explicit formulas are impossible. However, the meromorphic continuation
of $Z_{\varphi }^{(N)}\left( D;\boldsymbol{s}\right) $ is a linear
combination (the coefficients are holomorphic functions in the variables $s_{ij}$) of  Gamma functions evaluated at linear combinations of the
variables $s_{ij}$. For general functions $\varphi $, the computation of the
holomorphic coefficients is a difficult task, but the description of the
polar locus can be obtained algorithmically from a suitable desingularization
of a hypersurface like $A_{N}(x)$.

\smallskip
We now briefly situate this work within the general theory of local
zeta functions. Let $\mathbb{K}$ be a local field of characteristic zero,
for instance $\mathbb{R}$, $\mathbb{C}$ or $\mathbb{Q}_{p}$, the field of $p$-adic numbers. Set $\boldsymbol{f}:=\left( f_{1},\ldots ,f_{m}\right) $ and $\boldsymbol{s}:=\left( s_{1},\ldots ,s_{m}\right) \in \mathbb{C}^{m}$, where
the $f_{i}(x)$ are non-constant polynomials in the variables $x:=(x_{1},\ldots ,x_{n})$ with coefficients in $\mathbb{K}$. The
multivariate local zeta function attached to $(\boldsymbol{f},\Theta )$,
with $\Theta $  a test function (with compact support), is defined as 
\begin{equation*}
Z_{\Theta }\left( \boldsymbol{f},\boldsymbol{s}\right) =\int\limits_{\mathbb{
R}^{n}}\Theta \left( x\right) \prod\limits_{i=1}^{m}\left\vert
f_{i}(x)\right\vert ^{s_{i}}dx,\qquad \text{ with}\operatorname{Re}(s_{i})>0\text{
for all }i\text{,}
\end{equation*}
where $dx$ is the normalized Haar measure on $(\mathbb{K}^{n},+)$. These
integrals admit meromorphic continuations to the whole $\mathbb{C}^{m}$, 
\cite{Igusa-old}-\cite{Igusa}, \cite{Loeser},  \cite{Kashiwara-Takai}-\cite
{Zuniga-Veys2}. 

In the case $\mathbb{K}=\mathbb{R}$, $m=1$, these local zeta functions were
introduced in the 50s by Gel'fand and Shilov. The main motivation was that
the meromorphic continuation of Archimedean local zeta functions implies the
existence of fundamental solutions (i.e., Green functions) for differential
operators with constant coefficients, \cite{Atiyah}, \cite{Ber}. The regularization of Feynman amplitudes in quantum field theory is based on the analytic continuation of distributions attached to complex powers of polynomial functions in the sense of Gel'fand and Shilov \cite{G-S}, see also  \cite{B-G-Gonzalez-Dom}-\cite{Bogner}, among others. A new approach to studying scattering amplitudes, called positive geometries, has emerged, \cite{BEPV}-\cite{Arkani et al}. Some constructions in \cite{BEPV} look similar to those here, but in general terms, the two approaches are different. Finally,  the integrals $Z_{\varphi }^{(N)}\left( D;\boldsymbol{s}\right) $ can be
expanded as finite sums of multivariate local zeta functions.

\smallskip
The paper is organized as follows. In Section \ref{Prelim}, we review
Hironaka's resolution of singularities theorem, and we give some examples of embedded resolution of singularities for hyperplane arrangements that are relevant in our setting. We also
review some results about local zeta functions in general and some previous results about Koba-Nielsen local zeta functions. Section \ref
{Section 4} contains the main results and proofs of this paper, Theorem \ref
{maintheorem} and Proposition \ref{independence}. The proof of this
proposition involves elementary but technical calculations and is postponed
to the end of the section. In Section  \ref{Section 6}, we show how our main
result allows us to recover several main results in the paper by Sussman \cite
{Sussman}; these are, in turn, connected to genus zero open string amplitudes, as in the work of Brown and Dupont \cite{Brown-Dupont}.
In Section \ref{generalizations}, we consider a further
generalization of the Koba-Nielsen local zeta functions to arbitrary
hyperplane arrangements. In fact, our proofs in Section \ref{Section 4} are geometrically conceptual and more widely applicable, yielding Theorem \ref{thm general hyperplanes} and Proposition  \ref{prop general hyperplanes} in this more general context. 
 In the final section, we give for the interested reader a quick review of
several types of integrals, including Selberg-Mehta-Macdonald and
Dotsenko-Fateev-like integrals, and local zeta functions for graphs, 
which serve as motivation for the integrals studied here.

\section{Preliminaries}\label{Prelim}
In this paper, we work exclusively with real analytic manifolds. In
particular, $\mathbb{R}^{N}$ and $(\mathbb{P}_{\mathbb{R}}^{1})^{N}$, with $\mathbb P^1_\mathbb{R}$ the real projective line, are
considered $\mathbb{R}$-analytic manifolds. Then terms like local, manifold, etc. should be understood in the category of the $\mathbb{R}$-analytic manifolds.  Whenever we use {\em dimension}, it is the standard notion for real manifolds (possibly with boundary).  This is compatible with the notion of dimension for linear subspaces of (real) affine or projective spaces. Also, {\em codimension} always means codimension as subset of $\mathbb{R}^N$ or $(\mathbb P^1_\mathbb{R})^N$. 

\subsection{Embedded resolution of singularities}\label{resolution}

We state below Hironaka's embedded resolution theorem for a hypersurface $H$ on a general nonsingular variety $M$ over $\mathbb{R}$. In the core of the paper, $M$ will be either just $\mathbb{R}^N$ or $(\mathbb{P}_{\mathbb{R}}^{1})^{N}$, and $H$ will simply be a union of hyperplanes.
But we prefer to formulate this important theorem in the usual  general setting, in particular to be able to describe the general context of zeta functions in Section \ref{zetafunctions}.
 

\begin{theorem}
[Hironaka, \cite{H}]\label{thresolsing} 
Let $M$ be an $n$-dimensional nonsingular variety over $\mathbb{R}$, and $H$ a hypersurface on $M$. 
There exists an embedded resolution
$\pi:X\rightarrow M$ of $H$, that is,

\noindent(i)
$X$ is an $n$-dimensional nonsingular algebraic variety,
$\pi$ is a proper algebraic morphism, which is an isomorphism
outside of $\pi^{-1}(H)$, and  which can be constructed as a composition of a
finite number of blow-ups along closed nonsingular subvarieties, where all centres of blow-up, and hence also $X$ and $\pi$, are defined over $\mathbb{R}$ (in particular
$X$ and all centres of blow-up can also be considered as $\mathbb{R}$-analytic manifolds);

\noindent(ii) $\pi^{-1}\left(  H\right)  $ is a normal-crossing divisor,
meaning that $\pi^{-1}\left(  H\right)  =\cup_{i\in T}E_{i}$, where the
$E_{i}$\ are closed nonsingular subvarieties of $X$ of codimension one, intersecting transversally. That is,
at every point $b$ of $X$, there exist local coordinates $\left(  y_{1},\ldots,y_{n}\right)  $ on $X$ around $b$ such that, if $E_{1},\ldots,E_{r}$
are the $E_{i}$ containing $b$, we have on some open neighborhood  of $b$
that $E_{i}$ is given by $y_{i}=0$ for $i\in\{1,\ldots,r\}$.
\end{theorem}

\begin{remark}\label{explain}\rm
We warn the reader about the fact that Theorem \ref{thresolsing} is an existence theorem. An embedded resolution is not unique, and, although nowadays various algorithmic versions are known, constructing a resolution  of a given specific hypersurface $H$ is in general very cumbersome. 

On the other hand, for the hyperplane configurations $H$ that we study in the present paper, relatively easy and concrete algorithmic embedded resolutions are known.
Below we recall from \cite{BVZ-JHP} the concrete embedded resolution of the specific hyperplane arrangement given by (\ref{A_N}); it will be a crucial ingredient in the study of Koba-Nielsen zeta functions in Sections \ref{Koba-Nielsen-Section} and \ref{Section 4}.
In Section \ref{generalizations}, we use the more general version of this algorithm in the context of arbitrary hyperplane arrangements, as constructed in \cite{STV}-\cite{Va}. 
\end{remark}

We mention as a side remark that 
Hironaka's resolution theorem is more generally valid over any field of characteristic zero,
in particular over the local fields $\mathbb{R}$, $\mathbb{C}$, the field of
$p$-adic numbers $\mathbb{Q}_{p}$, or a finite extension of $\mathbb{Q}_{p}$.
 For a
discussion on the basic aspects of embedded resolutions in the context of applications to integrals, the reader may consult \cite[Chapter 2]{Igusa}.

\medskip
In the statement of Theorem \ref{thresolsing}, there are two kinds of subvarieties $E_{i},i\in T$. Each blow-up creates an
\textit{exceptional variety} $E_{i}, i\in T_e  \ (\subset T)$; the image by $\pi$ of  any of these
$E_{i}$ has codimension at least two in $\mathbb{R}^{n}$ or $M$. The other $E_{i}, i\in T_s \  (\subset T)$
are the so-called \textit{strict transforms} of the irreducible components of $H$.
More precisely, the strict transform of a component $H_i$ of $H$ is the closure (in $X$) of $\pi^{-1}(H \setminus \cup_{i\in T_e} \pi(E_i))$.

\medskip
The hypersurfaces we consider in the present paper are a special kind of hyperplane arrangements in $\mathbb{R}^N$ and $(\mathbb{P}_{\mathbb{R}}^1)^N$, respectively.
Let
\begin{equation}
f_N(x) :=
{\displaystyle\prod\limits_{i=1}^{N}}
x_{i}\text{ }
{\displaystyle\prod\limits_{i=1}^{N}}
\left(  1-x_{i}\right)  \text{ }
{\displaystyle\prod\limits_{1\leq i<j\leq N}}
\left(  x_{i}-x_{j}\right)   , \label{Eq_0}
\end{equation}
and $A_N:= f_N^{-1}(0)$ the induced affine hyperplane arrangement in $\mathbb{R}^N$.
There is an \lq economic\rq\ embedded resolution $\pi:X\rightarrow\mathbb{R}^{N}$ of $A_N$, obtained by blowing up all intersections $Z_j$ of components of $A_N$ that contain the points $\underline{0}=(0,0,\dots,0)$ or $\underline{1}=(1,1,\dots,1)$. Up to permutation of the coordinates, these blow-up centres are precisely the affine subspaces $Z_j$ given by
\begin{equation}\label{centres}
\begin{array}
[c]{l}
(a) \quad x_1 =\dots=x_r=0  \ (2\leq r \leq N),  \text { or } \\
(b) \quad  x_1=\dots=x_r=1  \ (2\leq r \leq N), \text{ or } \\
(c) \quad x_1 =x_2\dots=x_r  \ (3\leq r \leq N).
\end{array}
\end{equation}
More precisely, one blows up first the centres of dimension $0$ (which are just $\underline{0}$ and $\underline{1}$). After these two blow-ups,  the (transforms of the) centres of dimension $1$ become disjoint, and next one blows up  these centres.  One continues this way, blowing up centres of increasing dimension, ending with centres of dimension $N-2$. The embedded resolution $\pi$ is the composition of all these blow-ups.

Note that, again up to permutation of the coordinates, the $N+N+\frac{N(N-1)}{2}$ irreducible components  of $A_N$ are of the same form as (\ref{centres}), where $r=1$, $r=1$ and $r=2$ in (a), (b) and (c), respectively. We 

In the sequel, we will denote both these components of $A_N$ and the blow-up centres in (\ref{centres}) by $Z_j$. Thus, the strict transforms of the components $Z_j$ of $A_N$ are the $E_j, j\in T_s$, and the $Z_j$ from (\ref{centres})  are the images by $\pi$ of the $E_j, j\in T_e$.

One easily verifies that then, all together, this procedure results in $3\cdot 2^N -N -3$ components $E_j, j\in T,$ of $\pi^{-1}(A_N)$, or, equivalently, $3\cdot 2^N -N -3$ subspaces $Z_j, j\in T$.

\begin{example}\rm
When $N=2$, there are $5$ components of $A_2$, drawn in Figure 1.  We only blow up  the points $\underline{0}$ and $\underline{1}$, yielding in total $7$ subspaces $Z_j$.
\end{example}

\medskip
\centerline{
\beginpicture
\setcoordinatesystem units <.35truecm,.35truecm>

\putrule from -2 0 to 12 0
\putrule from -2 10 to 12 10
\putrule from 0 -2 to 0 12
\putrule from 10 -2 to 10 12

\setlinear    \plot  -2 -2   12 12  /


\put {$x_1=x_2$} at 5.3 4.5
\put {$x_1=0$} at -2 5
\put {$x_1=1$} at 12 5
\put {$x_2=0$} at 5 -.9
\put {$x_2=1$} at 5 10.9
\put{$\bullet$} at 0 0
\put{$\bullet$} at 10 10
\put{$\underline{0}$} at -0.7 0.7
\put{$\underline{1}$} at 10.7 9.3

\put{Figure 1} at 5 -3.5

\setcoordinatesystem units <.35truecm,.35truecm> point at -24 -3.5

\putrule from 8 8 to 0 8
\putrule from 8 8 to 8 0
\putrule from -3 -5 to -3 3
\putrule from -3 -5 to 5 -5
\linethickness=1.5pt
\putrule from 0 0 to 0 8
\putrule from 0 0 to 8 0
\putrule from 5 3 to -3 3
\putrule from 5 3 to 5 -5
\setlinear    \plot  0 8   -3 3 /
\setlinear    \plot  8 0   5 -5 /

\setplotsymbol ({$\cdot$})
\setlinear    \plot  0 0    5 3 /
\setlinear    \plot  0 0   -3 -5 /
\setlinear    \plot 5 3   8 8 /


\put{$\bullet$} at 0 0
\put{$\bullet$} at 5 3
\put{$\underline{0}$} at -0.7 0.7
\put{$\underline{1}$} at 4.3 3.7

\put{Figure 2} at 2 -7

\endpicture
}

\medskip
\begin{example}\rm
When $N=3$, there are $9$ components of $A_3$, namely $\{x_i=0\}, \{x_i=1\} (i=1,2,3)$ and $\{x_1=x_2\}, \{x_1=x_3\}, \{x_2=x_3\}$.  We first blow up  $\underline{0}$ and $\underline{1}$, and then (the transforms of) the $7$ lines $\{x_i=x_j=0\}, \{x_i=x_j=1\} (i\neq j)$ and $\{x_1=x_2=x_3\}$, indicated in bold in Figure 2. So in total we have $18$ subspaces $Z_j$.
\end{example}

\bigskip
In order to study the integrals (\ref{zeta_function_string}) over unbounded domains, one considers for instance the compactification $(\mathbb{P}_{\mathbb{R}}^1)^N$ of $\mathbb{R}^N$, and the induced hyperplane arrangement $\bar{A}_N$ in $(\mathbb{P}_{\mathbb{R}}^1)^N$, consisting of the closures of the components of $A_N$ in $(\mathbb{P}_{\mathbb{R}}^1)^N$, together with the $N$ \lq hyperplanes at infinity\rq. More precisely, taking $z_i$ as \lq coordinate at infinity\rq\ on $\mathbb{P}_{\mathbb{R}}^1$ (thus $z_i= 1/x_i$), the $N$ new hyperplanes are locally given by $z_i=0 \ (1\leq i \leq N)$.  On the other hand, the closures of the ones in $A_N$ given by $x_i=1$ and $x_i=x_j$ acquire the local descriptions $z_i=1$ and $z_i=z_j \ (i<j)$, respectively, at infinity.

Let us denote the point $\{z_1=z_2=\dots=z_N=0\}$ by $\underline{\infty}$. The components of $\bar{A}_N$ containing it are precisely $z_i =0 \  (1\leq i\leq N)$ and  $z_i=z_j \ (1\leq i<j\leq N)$. This is exactly the same local description as for the components containing $\underline{0}$. Hence, to construct
a similar \lq economic\rq\ embedded resolution $\pi:\bar{X}\rightarrow (\mathbb{P}_{\mathbb{R}}^1)^N$  of $\bar{A}_N$, one only needs to blow up also all intersections $Z_j$ of components of $\bar{A}_N$ that contain the point $\underline{\infty}$.
Now one verifies that, all together, there are $2^{N+2} -N -4$ components $E_i, i\in T,$ of $\pi^{-1}(\bar{A}_N)$, or, equivalently, $2^{N+2} -N -4$ subspaces $Z_j, j\in T$.
\smallskip
It is exactly this embedded resolution that was used in \cite{BVZ-JHP} to study Koba-Nielsen zeta functions.

\medskip
\begin{example}\rm
For $N=2$, there are two extra components in $\bar{A}_2$, namely $\{z_1=0\}$ and $\{z_2=0\}$, whose intersection is $\underline{\infty}$, as sketched in Figure 3. In order to construct $\bar{X}$, we perform one extra blow-up at $\underline{\infty}$. Now we have in total $7+3=10$ subspaces $Z_j$.
\end{example}

\medskip
\centerline{
\beginpicture
\setcoordinatesystem units <.35truecm,.35truecm>

\putrule from -2 0 to 12 0
\putrule from 0 -2 to 0 12

\putrule from -2 6 to 12 6
\putrule from 6 -2 to 6 12

\setlinear    \plot  -1.8 -1.8   11.6 11.6  /

\setdashes
\putrule from 10 -2 to 10 12
\putrule from -2 10 to 12 10


\put {$z_1=0$} at 12 4
\put {$z_2=0$} at 3 10.9
\put{$\bullet$} at 0 0
\put{$\bullet$} at 6 6
\put{$\bullet$} at 10 10
\put{$\underline{0}$} at -0.7 0.7
\put{$\underline{1}$} at 6.7 5.3
\put{$\underline{\infty}$} at 10.8 9.3

\put{Figure 3} at 5 -3.5
\endpicture
}

\medskip

\begin{example}\rm$\mathfrak{M}_{0,N+3}$
For $N=3$, there are three extra components in $\bar{A}_3$, namely $\{z_i=0\} (i=1,2,3)$, whose intersection is $\underline{\infty}$. In order to construct $\bar{X}$, we perform an extra blow-up at $\underline{\infty}$, and then three extra blow-ups with centre (the transforms of) the lines $\{z_i=z_j=0\} (i\neq j)$. Now we have in total $18+7=25$ subspaces $Z_j$.
\end{example}

\begin{remark} \rm 
It is well known that the complement 
$\mathbb{R}^N \setminus A_N = (\mathbb{P}_{\mathbb{R}}^1)^N \setminus \bar{A}_N$ can be viewed as the moduli space $\mathfrak{M}_{0,N+3}$ of rational curves with $N+3$ marked points.
The nonsingular projective variety  $\bar{X}$, constructed above, can in fact be viewed as the Deligne-Mumford compactification $\overline{\mathfrak{M}}_{0,N+3}$ of $\mathfrak{M}_{0,N+3}$ (using stable marked curves), and, equivalently, to the \lq dihedral compactification\rq\ in \cite[Section 2]{Brown}.

The latter is constructed via an embedding of $\mathfrak{M}_{0,N+3}$ in an appropriate power of $\mathbb{P}_{\mathbb{R}}^1\setminus \{0.1,\infty\}$, using cross-ratios. It is covered by  \lq dihedral\rq\ affine charts, associated to the various dihedral structures on $\{1,\dots,N+3\}$.
In contrast, $\bar{X}$ is covered by less affine charts, which are all isomorphic to an affine space. For example, when $N=2$, there are $12$ dihedral charts, while the construction of 
 $\bar{X}$, by three point blow-ups on $(\mathbb{P}_{\mathbb{R}}^1)^2$, naturally induces a covering by $8$ affine planes.
 Dihedral charts and coordinates are canonical, while our affine space charts and their coordinates  depend on some (quite natural) choices.
 \end{remark}

\begin{remark} \rm  
There is a canonical embedded resolution for hyperplane arrangements, the so-called \emph{wonderful resolution}, obtained by blowing up \emph{all} possible intersections of components \cite{DP}.  In order to compare efficiency, in the case of $A_N$, when constructing $\pi$ above, we only blow up two points,  $\underline{0}$ and $\underline{1}$, while, for the wonderful resolution, one blows up all the $2^N$ points with coordinates $0$ or $1$.
\end{remark}

\medskip
\subsection{Multivariate zeta functions}\label{zetafunctions}

Let $f_{1}(x),\ldots,f_{m}(x)\in\mathbb{R}\left[  x_{1},\ldots,x_{n}\right]  $
be non-constant polynomials and $H:=\cup_{i=1}^{m}f_{i}^{-1}(0)$. We set $\boldsymbol{f}:=\left(  f_{1},\ldots
,f_{m}\right)  $ and $\boldsymbol{s}:=\left(  s_{1},\ldots,s_{m}\right)
\in\mathbb{C}^{m}$.
We consider here some domain of integration $D\subset \mathbb{R}^n$ of dimension $n$, that is determined by linear inequalities; we include the case \lq no inequalities\rq, being $D=\mathbb{R}^n$.

For each function $\Theta:\mathbb{R}^{n}\rightarrow\mathbb{C}$ that is smooth on $D$, the multivariate local zeta function attached to
$(D,\boldsymbol{f},\Theta)$ is defined as
\begin{equation}
Z_{\Theta}\left(D,  \boldsymbol{f};\boldsymbol{s}\right)  =\int
\limits_{D} \Theta\left(  x\right) \prod
\limits_{i=1}^{m}\left\vert f_{i}(x)\right\vert ^{s_{i}}
dx. \label{Zeta_Function}
\end{equation}
When $D$ is bounded or when $\Theta$ has compact support, it is well known that $Z_{\Theta}\left( D, \boldsymbol{f};\boldsymbol{s}\right)$ converges and is holomorphic when $\operatorname{Re}(s_{i})>0$ for all $i$. Furthermore, it admits a meromorphic
continuation to the whole $\mathbb{C}^{m}$, see \cite{Atiyah}, \cite{Loeser}. The case $D=\mathbb{R}^{n}$ has been studied intensively, \cite{Igusa}, \cite{Igusa-old}, \cite{G-S}, \cite{AVG}, \cite{Ber}, \cite{Veys-Zuniga-Advances}. By applying Hironaka's resolution of singularities theorem
to $H$, the study of integrals of type (\ref{Zeta_Function}) is reduced to the
case of monomial integrals, which can be studied directly, see e.g.
\cite{Loeser}, \cite{Igusa}, \cite[Chap. II, \S \ 7, \ Lemme 4]{AVG}, \cite[Lemme 3.1]{D-S},
\cite[Chap. I, Sect. 3.2]{G-S}, and \cite[Lemma 4.5]{Igusa-old}.  When $D$ is not bounded and when $\Theta$ does not have compact support, one can consider for instance the compactification $(\mathbb{P}_{\mathbb{R}}^1)^N$ of $\mathbb{R}^N$, and use an embedded resolution of the closure of $H$ in $(\mathbb{P}_{\mathbb{R}}^1)^N$. Then one typically takes some appropriate (finite) partition of unity $(\rho_\ell(x))_\ell$ on $(\mathbb{P}_{\mathbb{R}}^1)^N$, where each $\rho_\ell$ has compact support. Then
$$
Z_{\Theta}\left( D, \boldsymbol{f};\boldsymbol{s}\right) = \sum_\ell
\int\limits_{D} \rho_\ell(x) \Theta\left(  x\right) \prod
\limits_{i=1}^{m}\left\vert f_{i}(x)\right\vert ^{s_{i}}
dx,
$$
reducing the situation to smooth functions $\Theta$ with compact support, see \cite{BVZ-JHP}.  Note that, even when each integral in the sum above converges and is holomorphic in some nonempty open domain of $\mathbb{C}^m$, this is not necessarily the case for $Z_{\Theta}\left( D, \boldsymbol{f};\boldsymbol{s}\right)$, since the intersection of these domains can be empty.

We recall some facts about monomial integrals. The simplest case is the one-dimensional integral, see e.g. \cite[Chap. I, Sect. 3.2-3.3]{G-S}.

\begin{lemma}[\protect Gel'fand-Shilov] \label{basic}

Let $\varphi:\mathbb{R} \to \mathbb{C}$ be a smooth function with compact support. The integrals
$$
J_+(s):=\int_{\mathbb{R}_{\geq 0}} \varphi(x)  |x|^s  dx   \qquad\text{ and }\qquad J(s) := \int_\mathbb{R} \varphi(x) |x|^s  dx,
$$
viewed as functions of the complex variable $s$,
converge and are holomorphic in the region $\operatorname{Re(s)}>-1$, and have a meromorphic continuation to the whole complex plane with possible poles of order 1.

More precisely, if one considers the integrals as distributions in $\varphi$, then $J_+(s)$ and $J(s)$ have poles at all negative integers and at all odd negative integers, respectively.
In particular, $-1$ is a pole in both cases.
\end{lemma}

A higher-dimensional and multivariate version, presented in a useful way for the sequel, is as follows.
The proof  is similar to the one given in \cite[Lemma 4.5]{Igusa-old}.

\begin{lemma}
\label{Lemma0}
Let $\varphi:\mathbb{R}^n \to \mathbb{C}$ be a smooth function with compact support.
Consider the integrals
\[
J_+(\boldsymbol{s})=  J_+(s_{1},\ldots,s_{m}):=
{\textstyle\int\limits_{\{y_1\geq 0\}}} \varphi(y)
\Phi\left(  y;s_{1},\ldots,s_{m}\right) \left\vert
y_{1}\right\vert ^{\sum_{j=1}^{m}a_{{j}}s_{j}+b-1}
dy,
\]
and
\[
J(\boldsymbol{s})=J(s_{1},\ldots,s_{m}):=
{\textstyle\int\limits_{\mathbb{R}^{n}}} \varphi(y)
\Phi\left(  y;s_{1},\ldots,s_{m}\right) \left\vert
y_{1}\right\vert ^{\sum_{j=1}^{m}a_{{j}}s_{j}+b-1}
dy\text{, }
\]
where the $a_{j}$ are integers (not all zero)
and $b$ is an integer, and $\Phi\left(  y;s_{1},\ldots,s_{m}\right)  $ is a smooth function in $y$, non-vanishing on the support of $\varphi$, and  which is holomorphic in the complex parameters $s_{1},\ldots,s_{m}$.
Then the following assertions hold:

\noindent(i)  
the integrals  $J_+(\boldsymbol{s})$ and $J(\boldsymbol{s})$ are convergent and define  holomorphic functions in the domain
 \[
\mathcal{R}:=
\left\{  (s_{1},\ldots,s_{m})\in\mathbb{C}^{m} \mid \sum_{j=1}^{m}a_{j}
\operatorname{Re}(s_{j})+b>0\right\}  ;
\]
\noindent(ii) 
they admit an analytic continuations to the whole $\mathbb{C}^{m}$, as a meromorphic functions with poles belonging to
\[
\bigcup_{t\in\mathbb{N}}\left\{  \sum_{j=1}
^{m}a_{j}s_{j}+b+t=0\right\};
\]

\noindent (iii) considering $J_+(\boldsymbol{s})$ and $J(\boldsymbol{s})$  as distributions in $\varphi$, the hyperplane $\sum_{j=1}^{m}a_{j}s_{j}+b=0$ belongs to the polar locus of both analytic continuations.
\end{lemma}


\medskip

The integral (\ref{Zeta_Function}) is studied by using the \lq change of variables\rq\ induced by a resolution $\pi$ as in Theorem \ref{thresolsing}. In a {\em generic} point of any $E_i, i\in T,$ the pullback of $\Theta\left(  x\right) \prod\limits_{i=1}^{m}\left\vert f_{i}(x)\right\vert ^{s_{i}}dx$ can be described in local coordinates $y$ as
\begin{equation}\label{integrand}
\varphi(y)\Phi\left(  y;s_{1},\ldots,s_{m}\right) \left\vert y_{1}\right\vert ^{\sum_{j=1}^{m}a_{{j}}s_{j}+b-1}dy,
\end{equation}
as in Lemma \ref{Lemma0}, where $E_i$ is locally given by $y_1=0$.

\begin{remark}\rm
(1) When $E_i$ is \lq affine\rq\ in the sense that $\pi(E_i)$ has non-empty intersection with $\mathbb{R}^N$, then  in (\ref{integrand}) all $a_j$ are nonnegative and $b$ is positive. This is always the case when $D$ is bounded or $\Theta$ has compact support, since then we only need a resolution $\pi$ of $H \subset \mathbb{R}^N$. 

(2) The study of the convergence, analytic continuation, and polar locus of the function (\ref{Zeta_Function}) is then essentially reduced to the study of finitely many integrals with an integrand of the form (\ref{integrand}). The (global) integration domain is the strict transform of $D$ by $\pi$; locally it looks typically like the integration domains $\{y_1\geq 0\}$ or $\mathbb{R}^N$ in Lemma \ref{Lemma0} or some set disjoint from $\{y_1= 0\}$.
\end{remark}

\subsection{\label{Koba-Nielsen-Section} Koba-Nielsen local zeta functions}

When $D={\mathbb{R}}^{n}$ and $\varphi \equiv 1$, the integral $Z_{\varphi
}^{(N)}(D;\boldsymbol{s})$ is a {\em Koba-Nielsen local zeta function}, denoted as $Z_{\mathbb{R}}^{(N)}\left( \boldsymbol{s}\right) $, i.e.,
\begin{equation}
Z_{\mathbb{R}}^{(N)}\left( \boldsymbol{s}\right) :=\int\limits_{\mathbb{R}
^{N}}\prod\limits_{i=1}^{N}\left\vert x_{i}\right\vert ^{s_{0i}}\left\vert
1-x_{j}\right\vert ^{s_{i(N+1)}}\text{ }\prod\limits_{1\leq i<j\leq
N}\left\vert x_{i}-x_{j}\right\vert ^{s_{ij}}dx,
\label{zeta_funtion_string}
\end{equation}
where $dx=\prod\nolimits_{i=1}^{N}dx_{i}$  denotes the Haar measure on $\mathbb{R}^{N}$.  In \cite{BVZ-JHP}, we showed that the Koba-Nielsen local zeta function $Z_{\mathbb{R}}^{(N)}\left( \boldsymbol{s}\right)$ is a linear combination of multivariate local zeta functions with test functions {\em having compact
support}.  Each of these local zeta functions is holomorphic in some open subset of ${\mathbb{C}}^{\boldsymbol{d}}$ (depending  on the zeta function), and admits a meromorphic continuation to the whole ${\mathbb{C}}^{\boldsymbol{d}}$.
Using the embedded resolution $\pi $ of $\bar{A}_{N}$ described above, we showed that all these local zeta functions are
convergent and holomorphic in a nonempty (open) {\em common} domain, determined by 
$2^{N+2}-N-4$ inequalities, coming from that same number of components 
$E_{i},i\in T,$ that arise in the resolution $\pi $, \cite[Proposition 5.3]{BVZ-JHP}. 
We then concluded that  $Z_{\mathbb{R}}^{(N)}\left( \boldsymbol{s}\right) $ itself is also
holomorphic in an open subset of ${\mathbb{C}}^{\boldsymbol{d}}$, and that it
admits a meromorphic continuation to the whole ${\mathbb{C}}^{\boldsymbol{d}}
$, \cite[Theorem 4.1]{BVZ-JHP}. 

\begin{remark}
One of the motivations for these meromorphic continuation results in  \cite{BVZ-JHP} was the regularize Koba-Nielsen open string amplitudes; we refer to  \cite[Section 1]{BVZ-JHP} for more information.
 In fact, since Hironaka's theorem is valid over any field of characteristic zero, we were able to regularize the Koba-Nielsen amplitudes
defined over $\mathbb{R}$, $\mathbb{C}$, or $\mathbb{Q}_{p}$, the field of $p
$-adic numbers, at the same time, \cite[Theorem 7.1]{BVZ-JHP}; see also \cite{BGZ-LMP}, \cite{Symmetry}, and the references therein. 
\end{remark}


We now state the explicit convergence and continuation results from \cite{BVZ-JHP}, noting that there the variables appearing in $Z_{\mathbb{R}}^{(N)}\left( \boldsymbol{s}\right)$ were labeled as $x_{2},\ldots ,x_{N-2}$, and here as $x_{1},\ldots ,x_{N}$.

\begin{theorem}[{\protect\cite[Theorem 4.1 and Proposition 5.3]{BVZ-JHP}}]
\label{conditions} The zeta function $Z_{\mathbb{R}}^{(N)}\left( \boldsymbol{s}\right) $
admits a meromorphic continuation to the whole ${\mathbb{C}}^{\boldsymbol{d}}
$, and the convergence conditions (and associated possible polar locus)
associated to each $E_{\ell }$ or $Z_{\ell }$, $\ell \in T,$ are as follows.


\noindent
Each of the $\frac{N(N+3)}{2}$ components $Z_\ell$  of $A_N$ induces the condition
\begin{equation}
\operatorname{Re}(s_{ij}) > -1
\label{EQ A}
\end{equation}
for the corresponding variable $s_{ij}$.

\noindent
The  exceptional $E_\ell$, coming from centres $Z_\ell$ of dimension $k\in\{0,\dots,N-2\}$  as in case (a) in (\ref{centres}), range over all subsets $J \subset \{1,\dots,N\}$ with $\sharp J =N-k$;  such an $E_\ell$ induces the condition
\begin{equation}
\sum_{j\in J}  \operatorname{Re}(s_{0j})  + \sum_{\substack{i,j \in J \\ i<j } }   \operatorname{Re}(s_{ij}) > - (N-k).
\label{EQ B}
\end{equation}

\noindent
The exceptional $E_\ell$, coming from  centres $Z_\ell$ of dimension $k\in\{0,\dots,N-2\}$ as in case (b) in (\ref{centres}), range over all subsets $J \subset \{1,\dots,N\}$ with $\sharp J =N-k$;   such an $E_\ell$ induces the condition
\begin{equation}
\sum_{j\in J}  \operatorname{Re}(s_{j(N+1)})  + \sum_{\substack{i,j \in J \\ i<j } }   \operatorname{Re}(s_{ij}) > - (N-k).
\label{EQ C}
\end{equation}

\noindent
The exceptional $E_\ell$, coming from  centres $Z_\ell$ of dimension $k\in\{1,\dots,N-2\}$ as in case (c) in (\ref{centres}), range over all subsets $J \subset \{1,\dots,N\}$ with $\sharp J =N-k+1$;  such an $E_\ell$ induces the condition
\begin{equation}
 \sum_{\substack{i,j \in J \\ i<j }}    \operatorname{Re}(s_{ij}) > - (N-k).
\label{EQ D}
\end{equation}

\noindent
The components $E_\ell$ \lq at infinity\rq, that is, with $\pi(E_\ell)=Z_\ell \subset (\mathbb{P}_{\mathbb{R}}^1)^N \setminus \mathbb{R}^N$ of dimension $d\in\{0,\dots,N-1\}$, range over all subsets $J \subset \{1,\dots,N\}$ with $\sharp J =N-d$; such an $E_\ell$ induces the condition
\begin{equation}
\sum_{j\in J}  \operatorname{Re}(s_{0j})+ \sum_{j\in J} \operatorname{Re}(s_{j\left(
N-1\right)} ) + \sum_{\substack{i \in \{1,\dots,N\}  \setminus J \\j\in J}}    \operatorname{Re}
(s_{ij}) + \sum_{\substack{i,j \in J \\ i<j }}    \operatorname{Re}(s_{ij}) < - (N-d).
\label{EQ E}
\end{equation}

In particular, this region is nonempty and it contains the concrete \lq hypercube\rq, given by
$$\frac{-2}{N+1}<\operatorname{Re}(s_{ij})<\frac{-2}{N+3} \qquad \text{  for all } ij.$$
\end{theorem}

\begin{remark}
(i)   In \cite{BVZ-JHP}, we did not address the question
whether these $2^{N+2}-N-4$ conditions  are independent;  this result is established in Proposition \ref{independence} below. 
Consequently, the hyperplanes in $\mathbb{C}^{\boldsymbol{d}}$, given by replacing \textquotedblleft $>$\textquotedblright\ by \textquotedblleft
$=$\textquotedblright\ in the conditions (\ref{EQ A}) till  (\ref{EQ E}),  {\em all} belong to the polar locus of $Z_{\mathbb{R}}^{(N)}\left( \boldsymbol{s}\right) $.

(ii) A crucial observation is that Theorem  \ref{conditions} is also
valid for arbitrary integrals $Z_{\varphi }^{(N)}(D;\boldsymbol{s})$. 
But when $Z_{\varphi }^{(N)}(D;\boldsymbol{s})\neq Z_{\mathbb{R}}^{(N)}\left( \boldsymbol{s}\right) $, typically {\em not all}  the conditions listed in Theorem \ref{conditions} give rise to poles. The main result of the present paper is a criterion determining the hyperplanes in (i) that belong to the polar locus of $Z_{\varphi }^{(N)}(D;\boldsymbol{s})$.
\end{remark}

\begin{remark}\rm
The form of the \lq affine\rq\  conditions (\ref{EQ A}),  (\ref{EQ B}),   (\ref{EQ C}) and (\ref{EQ D}) has a straightforward geometric meaning.  
The sum always runs over those $ \operatorname{Re} (s_{ij})$ that correspond to the components of $A_N$ that contain $Z_\ell$. This is maybe not immediately clear from our notational choice for the variables  $s_{ij}$, which is just one of the standard ones when studying integrals for $f_N$.
In Section \ref{generalizations}, we will generalize our main result to arbitrary hyperplane arrangements, where the more general notation will make this apparent.
\end{remark}

\medskip
\begin{example}\label{conditionsN=2}\rm
We list the $10$ convergence conditions when $N=2$. See Figure 3 for their geometric origin.   The \lq affine\rq\ conditions (\ref{EQ A}),  (\ref{EQ B})  and (\ref{EQ C}) are
\begin{equation}
\begin{aligned}
&\operatorname{Re}(s_{01}) > -1, \ \operatorname{Re}(s_{02}) > -1,\ \operatorname{Re}(s_{12}) > -1,\ \operatorname{Re}(s_{13}) > -1,\ \operatorname{Re}(s_{23}) > -1; \\
&\operatorname{Re}(s_{01}) + \operatorname{Re}(s_{02}) + \operatorname{Re}(s_{12}) > -2; \\
&\operatorname{Re}(s_{13}) + \operatorname{Re}(s_{23}) + \operatorname{Re}(s_{12}) > -2;
\end{aligned}
\end{equation}
respectively. Here, condition  (\ref{EQ D}) does not appear.  The conditions \lq at infinity\rq\  (\ref{EQ E}) are
\begin{equation}
\begin{aligned}
&\operatorname{Re}(s_{01}) + \operatorname{Re}(s_{12}) + \operatorname{Re}(s_{13}) < -1, \quad
\operatorname{Re}(s_{02}) + \operatorname{Re}(s_{12}) + \operatorname{Re}(s_{23}) <-1; \\
&\operatorname{Re}(s_{01}) +  \operatorname{Re}(s_{02}) + \operatorname{Re}(s_{12}) + \operatorname{Re}(s_{13}) + \operatorname{Re}(s_{23}) <-2.
\end{aligned}
\end{equation}
The $25$ convergence conditions for $N=3$ are listed explicitly in \cite[Example 5.6]{BVZ-JHP} (with another index convention).
\end{example}

\begin{remark}\label{Gamma}\rm
Whenever a concrete (finite) list of convergence conditions as above is known, one can if fact express $Z_\varphi^{(N)}(D;\boldsymbol{s})$ as a sum, where each term is a product of an entire holomorphic function and (at most $N$)  Gamma functions, as explained in \cite[8.2]{BVZ-JHP}. We then have as a consequence that $Z_\varphi^{(N)}(D;\boldsymbol{s})$ can be written as a product of an entire function and Gamma functions.
For instance,  using notation as in Theorem \ref{conditions}, we can write  $Z_{\mathbb{R}}^{(N)}\left( \boldsymbol{s}\right) $  as follows:
\begin{eqnarray*}\label{Gammasum}
Z_{\mathbb{R}}^{(N)}\left( \boldsymbol{s}\right) =&
 F^{(N)}(\boldsymbol{s})\prod_{k=0}^{N-1}\prod_{\dim Z_\ell =k} \Gamma \left(\sum  s_{ij} + (N-k)\right) \times \\
& \prod_{d=0}^{N-1} \prod_{\dim Z_\ell =d} \Gamma \left(-\sum  s_{ij} - (N-d)\right).
\end{eqnarray*}

Here $F^{(N)}(\boldsymbol{s})$ is an entire function, $\Gamma(\cdot)$ denotes the classical Gamma function, the first products range over the \lq affine\rq\ $Z_\ell$, the second products over the $Z_\ell$ \lq at infinity\rq,  and the sums range over exactly the variables $s_{ij}$ occurring in  (\ref{EQ A}), (\ref{EQ B}), (\ref{EQ C}), (\ref{EQ D}) or (\ref{EQ E}), depending on the concrete $Z_\ell$.
\end{remark}

Recall that we denote both the components of $A_N$ or $\bar{A}_N$ and the blow-up centres  by $Z_j$. Thus, the strict transforms of the components $Z_j$ of $A_N$ or $\bar{A}_N$ are the $E_j, j\in T_s$, and the blow-up centres $Z_j$  are the images by $\pi$ of the $E_j, j\in T_e$.

\section{\label{Section 4} The main results}

\begin{proposition}\label{independence}
The $2^{N+2} -N -4$ conditions in Theorem \ref{conditions} are independent. That is, for each of them we can find a point $Q$ in $\mathbb{C}^{\boldsymbol{d}}$ \emph{not} satisfying that condition, but satisfying all the other ones.
\end{proposition}

The challenge to prove Proposition \ref{independence} is to find such adequate points $Q$. Then the verification of the statement consists of various computations. We defer this technical issue to Section \ref{independenceproof}.

\medskip
We will determine the polar locus of the zeta function $Z_\varphi^{(N)}\left( D; \boldsymbol{s}\right)$ in terms of the embedded resolution $\pi$ of the associated affine or projective hyperplane arrangement. We recall some notation from Subsection \ref{resolution}:
\begin{itemize}
\item
$
f_N(x)=
\prod_{i=1}^{N} x_{i} \prod_{i=1}^{N} (  1-x_{i}) \prod_{1\leq i<j\leq N} (  x_{i}-x_{j})
$,
$A_N= f_N^{-1}(0)\subset \mathbb{R}^N$ is the induced affine hyperplane arrangement, and $\bar{A}_N\subset (\mathbb{P}_\mathbb{R}^1)^N$  the completed arrangement;
\item
$
Z_\varphi^{(N)} (D; \boldsymbol{s})  = \int_D \varphi(x) \prod_{i=1}^{N} \vert x_{i}\vert^{s_{0i}} \prod_{i=1}^{N}\vert 1-x_{i}\vert^{s_{i(N+1)}} 
\prod_{1\leq i<j\leq N} \vert x_{i}-x_{j}\vert^{s_{ij}} $

$\times   dx, $
where $\varphi:\mathbb{R}^N\to \mathbb{C}$ is smooth on $D$;

\item as integration domain $D$ we consider all possible (compact or noncompact) polyhedra with boundary conditions given by inequalities coming from components of $A_N$. Thus,  $D$ is given by conditions of the form
$x_i \geq 0$, $x_i \leq 0$, $x_i \geq 1$ or $x_i\leq 1$ for some indices $i$ and/or $x_i\geq x_j$ for some $i\neq j$, such that $\dim(D)=N$.
\end{itemize}
Also, in the mentioned section,  we  introduced the  embedded resolutions $X\rightarrow\mathbb{R}^{N}$ of $A_N$ and $\bar{X}\rightarrow (\mathbb{P}_\mathbb{R}^1)^{N}$ of $\bar{A}_N$. We will use the first one when $D$ is compact, and the second one when $D$ is noncompact.

When $D$ is noncompact, we rather consider its closure in $(\mathbb{P}_\mathbb{R}^1)^N$.  In that case, with the coordinates $z_i=1/x_i$ \lq at infinity\rq, the closure of $D$ is locally given \lq at infinity\rq\ by similar inequalities $z_i \geq 0$, $z_i \leq 0$, $z_i \leq 1$, $z_i \geq 1$ and $z_i \geq z_j$.  And then we must consider also \lq faces at infinity\rq.

For simplicity of notation, we will denote both embedded resolutions by $\pi$, the components of $\pi^{-1}(A_N)$ or $\pi^{-1}(\bar{A}_N)$ by $E_i, i\in T$, and, when $D$ is noncompact, we keep the notation $D$ for its closure in $(\mathbb{P}_\mathbb{R}^1)^N$.
 We also denote by $\tilde{D}$ the strict transform of $D$ in $X$ or $\bar{X}$, respectively, as well as at any stage of the blow-up process.

\begin{theorem}\label{maintheorem}
A subspace $Z_j$ contributes to the polar locus of $Z_\varphi^{(N)}\left( D; \boldsymbol{s}\right)$ if and only if $\dim(Z_j \cap D) = \dim (Z_j)$.
\end{theorem}
The theorem follows from Propositions \ref{if} and \ref{only if} below. With \lq contributes to the polar locus\rq\ we mean that, when $L(\boldsymbol{s}) > b$ or $L(\boldsymbol{s}) < b$ is the condition in Theorem \ref{conditions} associated to $Z_j$ (with thus $L(\boldsymbol{s})$ a sum of some variables $s_{ij}$ and $b$ a negative integer), then the hyperplane given by $L(\boldsymbol{s}) = b$ in $\mathbb{C}^{\boldsymbol{d}}$ is part of the polar locus of $Z_\varphi^{(N)}\left( D; \boldsymbol{s}\right)$.

\begin{example}\label{exampleDelta2}\rm
In Figure 4, we indicated the domain $D=\Delta_2$ as the dotted area.  The intersections of the lines  $\{x_1=0\}$, $\{x_2=1\}$ and $\{x_1=x_2\}$ with $D$ have dimension $1$, so those $Z_j$ contribute to the polar locus. On the other hand, the intersections of $\{x_1=1\}$ and $\{x_2=1\}$ with $D$ have dimension $0$, so those $Z_j$ do not contribute.
The intersections of the points $\underline{0}$ and $\underline{1}$ with $D$ are of course those points themselves, hence they contribute to the polar locus.
\end{example}

\medskip
\centerline{
\beginpicture
\setcoordinatesystem units <.35truecm,.35truecm>

\putrule from -2 0 to 12 0
\putrule from -2 10 to 12 10
\putrule from 0 -2 to 0 12
\putrule from 10 -2 to 10 12
\setlinear    \plot  -2 -2   12 12  /

\setlinear
\vshade 0 0 10    10 10 10 /


\put {$x_1=x_2$} at 5.5 4.5
\put {$x_1=0$} at -2 5
\put {$x_1=1$} at 12 5
\put {$x_2=0$} at 5 -1
\put {$x_2=1$} at 5 11
\put{$\bullet$} at 0 0
\put{$\bullet$} at 10 10
\put{$\underline{0}$} at -0.7 0.7
\put{$\underline{1}$} at 10.7 9.3

\put{Figure 4} at 5 -3.5

\setcoordinatesystem units <.35truecm,.35truecm> point at -20 0

\putrule from -2 0 to 12 0
\putrule from 0 -2 to 0 12

\putrule from -2 6 to 12 6
\putrule from 6 -2 to 6 12

\setlinear    \plot  -1.8 -1.8   12 12  /

\setdashes
\putrule from 10 -2 to 10 12
\putrule from -2 10 to 12 10

\setsolid

\setlinear
\vshade 0 0 10    10 10 10 /

\put {$z_1=0$} at 12 4
\put {$z_2=0$} at 3 10.9
\put{$\bullet$} at 0 0
\put{$\bullet$} at 6 6
\put{$\bullet$} at 10 10
\put{$\underline{0}$} at -0.7 0.7
\put{$\underline{1}$} at 6.7 5.3
\put{$\underline{\infty}$} at 10.8 9.3

\put{Figure 5} at 5 -3.5
\endpicture
}

\medskip
\begin{example}\rm
In Figure 5, we indicate (the closure in $(\mathbb{P}_\mathbb{R}^1)^2$ of) the unbounded domain $D=\{(x_1,x_2) \in \mathbb{R}^2 \mid 0\leq x_1 \leq x_2\}$ as the dotted area.
Note that the condition $0\leq x_1 \leq x_2$ induces the condition $0\leq z_2 \leq z_1$.

In addition to the contributions listed in   Example \ref{exampleDelta2}, here also the lines $\{x_1=1\}$, and $\{z_2=0\}$, and the point $\underline{\infty}$ contribute.
In other words, from the list of $10$ subspaces $Z_j$ in  $(\mathbb{P}_\mathbb{R}^1)^2$, only the lines $\{x_2=0\}$ and $\{z_1=0\}$  {\em do not} contribute.
\end{example}

\medskip

\begin{proposition}\label{if}
If $\dim(Z_j \cap D) = \dim (Z_j)$, then $Z_j$ contributes to the polar locus of $Z_\varphi^{(N)}\left( D; \boldsymbol{s}\right)$.
\end{proposition}

\begin{proof}
We first consider the easy case when $Z_j\cap \operatorname{Int}(D) \neq \emptyset$.  In such a case, $E_j \cap \operatorname{Int}(\tilde{D}) \neq \emptyset$, and $E_j$ contributes to the polar locus by Lemma \ref{Lemma0}.
From now on, we assume that $F_j:= Z_j \cap D$ is a face of $D$ and set $r:= \operatorname{codim}(Z_j)=\operatorname{codim} (F_j)$. In particular, if $Z_j$ is a component of $A_N$ or $\bar{A}_N$, then $F_j$ is a facet of $D$, and, if $E_j$ is exceptional, then $F_j$ is a face of $D$ of codimension at least 2.

\smallskip
Using the blow-up formula, we will describe $E_j \cap \tilde{D}$ in local coordinates. Recall that blowing up is a geometric notion, independent of the choice of coordinates to compute it. For our purposes, it is useful to perform  a change of coordinates $x\mapsto y$  in $\mathbb{R}^N$  or $(\mathbb{P}_\mathbb{R}^1)^N$, such that the new origin is a \emph{general} point $P$ of $F_j$ (which is of course just $F_j$ if $F_j$ is a point), and such that

(1) $Z_j$ is given by the equalities $y_1=\dots=y_r=0$;

(2)  $D$ is given (locally around $P$) by the inequalities
\begin{equation}\label{local D}
\left\{
\begin{array}
[c]{l}
y_1 \geq 0 \\
\cdots\\
y_r \geq 0 \\
L_1(y_1,\dots,y_r)\geq 0 \\
\cdots \\
L_t(y_1,\dots,y_r)\geq 0.
\end{array}
\right.
\end{equation}
Here $t\geq 0$ and the $L_i(y)$ are linear forms in $y_1,\dots,y_r$. If $r=1$ or $r=2$,  one can always find coordinates $y$ such that not all $L_i$ are needed. But for $r\geq 3$, they may be necessary; see, for instance, Example \ref{example L}.

\smallskip
Let $u_1, u_2, \dots, u_N$ be the coordinates in the \lq first chart\rq\ of the blow-up of $Z_j$; then $u$ and $y$ are related by
\begin{equation}\label{lblowup}
\left\{
\begin{array}
[c]{l}
y_1    = u_1\\
y_i   = u_1u_i  \,(2\leq i \leq r) \\
y_\ell  = u_\ell \, (r < \ell \leq N),
\end{array}
\right.
\end{equation}
and $E_j$ is given in this chart by $u_1=0$.
Furthermore, here $\tilde{D}$ is given by
\begin{equation}\label{local Dstrictj}
\left\{
\begin{array}
[c]{l}
u_1 \geq 0 \\
u_1u_2 \geq 0 \\
\cdots\\
u_1u_r \geq 0 \\
L_1(u_1,u_1u_2,\dots,u_1u_r)\geq 0 \\
\cdots \\
L_t(u_1,u_1u_2,\dots,u_1u_r)\geq 0
\end{array}
\right.
\qquad \Leftrightarrow \qquad
\left\{
\begin{array}
[c]{l}
u_1 \geq 0 \\
u_2 \geq 0 \\
\cdots\\
u_r \geq 0 \\
L_1(1,u_2,\dots,u_r)\geq 0 \\
\cdots \\
L_t(1,u_2,\dots,u_r)\geq 0.
\end{array}
\right.
\end{equation}
Note that $\tilde{D}$ has positive measure, since $D$ has positive measure. (When $t=0$ above, this is also clear by the concrete description in (\ref{local Dstrictj}).)

The essential information in the above description is that, while (in that chart) $E_j$ is given by $u_1=0$, the domain $\tilde{D}$ is given by $u_1\geq0$ and other inequalities involving only \emph{other} variables.
So we can conclude by Lemma \ref{Lemma0} and Proposition \ref{independence}.
\end{proof}

\begin{example}\label{example L}\rm
Take $N=3$ and let $D$ be the (noncompact) domain given by the inequalities $x_1\geq 0$, $x_2\geq 0$, $x_1\geq x_3$ and $x_2\geq x_3$. In particular, locally around $\underline{0}$, we need all four inequalities to describe $D$. We perform the linear change of coordinates
\begin{equation}
\left\{
\begin{array}
[c]{l}
y_1 =x_1   \\
y_2 = x_2   \\
y_3 = x_1-x_3.
\end{array}
\right.
\end{equation}
Then in the $y$-coordinates $D$ is given by the inequalities $y_1\geq 0$, $y_2\geq 0$, $y_3\geq 0$ and $-y_1+y_2+y_3\geq 0$.
\end{example}

Below, we will need to consider subspaces generated by faces of $D$. When $D\subset \mathbb{R}^N$ is compact, each face $F$ of $D$ is in $\mathbb{R}^N$ and we mean \emph{affine} subspace generated by $F$. Otherwise, the compactified $D$ and its faces $F$ lives in $(\mathbb{P}_\mathbb{R}^1)^N$, and we mean \emph{projective} subspace generated by $F$.

\begin{proposition}\label{only if}
If $\operatorname{dim}(Z_j\cap D) < \operatorname{dim}(Z_j)$, then $Z_j$ does not contribute to the polar locus of $Z_\varphi^{(N)}\left(D;  \boldsymbol{s}\right)$.
\end{proposition}

\begin{proof}
We denote  $F_j:= Z_j \cap D$. We assume that $F_j\neq\emptyset$, since otherwise the conclusion is obvious. Now we need also the  subspace $C_j$ generated by $F_j$, and we denote $r:=\operatorname{codim}(C_j)=\operatorname{codim}(F_j)$.  Note that in this case $r > \operatorname{codim}(Z_j)$.
We perform the same  change of coordinates $x\mapsto y$  as above such that the new origin is a \emph{general} point $P$ of $F_j$ and such that

(1) $C_j$ is given by the equalities $y_1=\dots=y_r=0$;

(2)  $D$ is given (locally around $P$) by the inequalities (\ref{local D}).

\noindent
But now  $Z_j$ \emph{strictly contains} $C_j$. In any case, it is given by a finite set of linear equations in $y_1,\dots, y_r$, say
$$\sum_{i=1}^r a_i^{(m)} y_i =0,  \qquad \text{for } m\in M.$$
The fact that $Z_j \cap D$ is precisely $F_j$ can be formulated as the following implication.

\emph{(*) If $(y_1,\dots, y_r)$ satisfies both
\begin{equation}\label{condition}
\left\{
\begin{array}
[c]{l}
\sum_{i=1}^r a_i^{(m)} y_i =0,  \quad \text{for } m \in M ,\\
y_1\geq 0,\dots,y_r\geq 0, L_1(y_1,\dots,y_r)\geq 0, \dots, L_t(y_1,\dots,y_r)\geq 0,
\end{array}
\right.
\end{equation}
then  $y_1 = \dots = y_r=0$.}

\smallskip
We will consider the blow-up of $C_j$. In case this blow-up is {\em not} one of the blow-ups of our embedded resolution $\pi$, we extend $\pi$ with this blow-up, at the stage when centres of $\operatorname{dim}(C_j)$ are handled.  Important to note here is that the a priori extra candidate polar hyperplane induced by (the exceptional component of) this blow-up in fact {\em does not} contribute to the actual polar locus of  $Z_\varphi^{(N)}\left(D;  \boldsymbol{s}\right)$, since we know that the polar locus is included in the list induced by Theorem \ref{conditions}.

In fact, the previous argument turns out to be unnecessary, since  $C_j$ is really a centre of blow-up of $\pi$. Although not needed for the present proof, this result is of independent interest, which we show for the interested reader in Lemma \ref{blowupcentre} below.

\smallskip
Let $X'$ denote the ambient space after the blow-up with centre $C_j$.
As above, this blow-up is described in the first chart by (\ref{lblowup}) and the strict transform  $\tilde{D}$ is given by (\ref{local Dstrictj}).
Now the strict transform $\tilde{Z_j}$ of $Z_j$ in $X'$ is given by the equalities
$$a_1^{(m)}+\sum_{i=2}^r a_i^{(m)} u_i =0,  \qquad \text{for }m \in M.$$

We claim that  $\tilde{Z_j}$ and $\tilde{D}$ are disjoint in $X'$.
Indeed, suppose that
\begin{equation*}
(u_{1},u_{2},\dots ,u_{N})\in \tilde{Z_{j}}\cap \tilde{D}.
\end{equation*}
Then $(1,u_2,\dots,u_N)$ satisfies (\ref{condition}).  But then the implication (*) yields that
$1=u_2=\dots = u_N=0$, providing a contradiction.

Since at this stage of the embedded resolution algorithm the strict transforms of $Z_j$ and $D$ are already disjoint, certainly also $E_j$ and $\tilde{D}$ are disjoint (and both closed) in $X$ or $\bar{X}$.  Hence $E_j$ does not contribute to the polar locus.
\end{proof}

\begin{lemma}\label{blowupcentre}
Let $E_j$ be an exceptional component or strict transform in the embedded resolution $\pi$  for which $Z_j\cap D\neq\emptyset$ and $\dim(Z_j \cap D) < \dim (Z_j)$. Then the subspace generated by $Z_j \cap D$ is the centre of blow-up in $\pi$.
\end{lemma}

\begin{proof}
 Without loss of generality, we may assume (after possibly a coordinate change in $\mathbb{P}_{\mathbb{R}}^1$ and a permutation of the coordinates)  that $Z_j=\pi(E_j)$ is given either by
$x_1=\dots=x_r=0$ ($1\leq r\leq N-1$) or $x_1=\dots=x_r$ ($2\leq r\leq N$).
We denote as before $F_j:=Z_j \cap D$ and $C_j$ for the affine subspace generated by $F_j$.  We have that $Z_j \supsetneq C_j$, $Z_j\cap D = C_j\cap D = F_j$ and hence $\operatorname{codim}(Z_j) < \operatorname{codim}(C_j) =\operatorname{codim}(F_j)$.

\smallskip
Suppose that $C_j$ is \emph{not} a centre of blow-up in $\pi$.  Then necessarily $C_j$ is given by equations

\noindent
(1) $x_1=\dots=x_m=x_{m+1}=\dots=x_n=0$  or $x_1=\dots=x_m=x_{m+1}=\dots=x_n$ ($m\leq n$), and

\noindent
(2) one or more other equations of the form $x_j=1$ or $x_j=x_k$ involving variables $x_j,x_k$ ($k>j>n$).

We view $C_j$ as $C_1 \cap C_2$, where $C_1$ and $C_2$ are the subspaces given by the equations (1) and (2), respectively.  Then (locally around $C_j$) the domain $D$ can also be described as $D_1 \cap D_2$, where $D_1$ and $D_2$ are given by inequalities  involving only the variables $x_1, \dots, x_n$ and only the variables $x_{n+1}, \dots, x_N$, respectively. Note that thus $Z_j \supset C_1 \supsetneq C_j$.

\smallskip
\noindent
{\sc Claim.} We have that $\operatorname{codim}(C_i) = \operatorname{codim}(C_i\cap D)$ for $i=1,2$.

\smallskip\noindent
We first show the claim, using

(i) $\operatorname{codim}(D)=\operatorname{codim}(D_1)=\operatorname{codim}(D_2)=0$;

(i) the fact that, when subsets $S_1$ and $S_2$ of $\mathbb{R}^N$ are given by equations or inequalities \emph{in disjoint variables}, then $\operatorname{codim}(A\cap B) = \operatorname{codim}(A) + \operatorname{codim}(B)$.

We have that
\begin{align*}
\operatorname{codim}(F_j)&=\operatorname{codim}(C_j \cap D) = \operatorname{codim}(C_j )=\operatorname{codim}(C_1)+\operatorname{codim}(C_2)   \\
&\leq \operatorname{codim}(C_1\cap D)+\operatorname{codim}(C_2\cap D) \\
&= \operatorname{codim}((C_1\cap D_1) \cap D_2)+\operatorname{codim}((C_2\cap D_2)\cap D_1) \\
&= \operatorname{codim}(C_1\cap D_1) +\operatorname{codim}( D_2)+\operatorname{codim}(C_2\cap D_2)+\operatorname{codim}( D_1) \\
&= \operatorname{codim}(C_1\cap D_1) +\operatorname{codim}(C_2\cap D_2) \\
&= \operatorname{codim}(C_1\cap D_1 \cap C_2 \cap D_2) \\
&=\operatorname{codim}(C_j \cap D)=\operatorname{codim}(F_j),
\end{align*}
and hence the inequality must also be an equality, yielding the claim.

\smallskip
Now, since $Z_j \supset C_1 \supsetneq C_j$, we have certainly that $Z_j\cap D \supset C_1\cap D \supset C_j\cap D$.  But since both the first and last set are just $F_j$, both inclusions are equalities.  We rewrite the last equality:
$$(C_1\cap D_1)\cap D_2 = (C_1\cap D_1)\cap (C_2 \cap D_2),$$
and derive that $\operatorname{codim}im(C_1\cap D_1) = \operatorname{codim}(C_1\cap D_1)+\operatorname{codim}(C_2\cap D_2)$ and thus that $\operatorname{codim}(C_2\cap D_2)=0$. But then the claim above for $i=2$ would imply that $\operatorname{codim}(C_2)=0$, yielding a contradiction since by assumption $\operatorname{codim}(C_2)>0$.
\end{proof}

\subsection{ Proof of Proposition \ref{independence} }\label{independenceproof}

We recall the statement. \emph{The $2^{N+2} -N -4$ conditions in Theorem \ref{conditions} are independent. That is, for each of them we can find a point $Q$ in $\mathbb{C}^{\boldsymbol{d}}$ \emph{not} satisfying that condition, but satisfying all the other ones.}


(i) We start by fixing an \lq affine\rq\ condition, induced by a subspace $Z_j$ with $\operatorname{dim}(Z_j)=k$, where $k\in\{0, \dots, N-1\}$.
 So, for $k=N-1$, $Z_j$ is a component of $A_N$;  apart from that, $Z_j$ is a $k$-dimensional centre of blow-up.
In this way, we treat simultaneously all conditions of the form (\ref{EQ A}), (\ref{EQ B}), (\ref{EQ C}), and (\ref{EQ D}).

We have that $Z_j$ is contained in $\binom{N+1-k}{2}$ components of $A_N$; we put all the $\binom{N+1-k}{2}$ associated $\operatorname{Re}(s_{ij})$ equal to
\begin{equation}\label{value1}
\frac{-2}{N+1-k}.
\end{equation}
Then, their sum equals $-(N-k)$, and hence the inequality associated with $Z_j$ is {\em not} satisfied.
We choose another value for all other $\operatorname{Re}(s_{ij})$, such that the total average is equal to $-\frac{2}{N+2}$.  A straightforward calculation yields that this value is
\begin{equation}\label{value2}
\frac{-2}{N+2}  \cdot \frac{(k+1)N+2k}{2(k+1)N - k(k-1)}.
\end{equation}
Note that the value (\ref{value1}) is smaller than the average and (hence) smaller than the value (\ref{value2}).

We must show that with these choices {\em all other} inequalities are satisfied. We first consider the \lq affine\rq\ inequalities (\ref{EQ A}), (\ref{EQ B}), (\ref{EQ C}) and (\ref{EQ D}), associated to some $Z$ different from $Z_j$.

\smallskip
\noindent {\sc Case 1 :}  $\dim Z = k$.
The sum of the $\binom{N+1-k}{2}$ constituting $\operatorname{Re}(s_{ij})$ is certainly larger than $-(N-k)$, since at least one summand has the value (\ref{value2}).

\smallskip
\noindent {\sc Case 2 :}  $\dim Z = \ell > k$.
It is sufficient to check the \lq worst case\rq, occurring when the $\binom{N+1-\ell}{2}$ summands  $\operatorname{Re}(s_{ij})$ are as small as possible, which in this case means they all have the value (\ref{value1}).  And indeed
$$
\binom{N+1-\ell}{2} \cdot \left[ \frac{-2}{N+1-k} \right] > -(N-\ell) ,
$$
being simply equivalent to $\ell >k$.

\smallskip
\noindent {\sc Case 3 :}  $\dim Z = \ell <k$.
Now the worst case occurs when $\binom{N+1-k}{2}$ summands have the smallest value (\ref{value1}) and the other $\binom{N+1-\ell}{2}-\binom{N+1-k}{2}$ summands have the value (\ref{value2}), and then we must check that
\begin{align*}
&\binom{N+1-k}{2} \cdot \left[ \frac{-2}{N+1-k} \right]  \\
+ &\left[ \binom{N+1-\ell}{2}-\binom{N+1-k}{2}\right] \cdot \frac{-2}{N+2}  \cdot \frac{(k+1)N+2k}{2(k+1)N - k(k-1)}  > -(N-\ell) .
\end{align*}
By a straightforward computation, this is equivalent to the valid inequality
$(k\ell + \ell + k+ 3)N + 2k\ell > 0$.

\medskip
Second, we consider the  inequalities \lq from infinity\rq\ (\ref{EQ E}), associated to a subspace $Z_j$ \lq at infinity\rq\  of dimension $d\in \{0,\dots, N-1\}$.
Now we have $\frac{(N-d)(N+d+3)}{2}$ summands. Note that now the worst cases occur when the summands  $\operatorname{Re}(s_{ij})$ are as large as possible, meaning choosing as many values (\ref{value2}) as possible.

 There will be two cases, depending on the maximum of $\frac{(N-d)(N+d+3)}{2}$ and
$$
\frac{(N-d)(N+d+3)}{2} - \binom{N+1-k}{2} = \frac{2(k+1)N-k(k-1)}{2}.
$$
We claim that
\begin{equation}\label{d-inequality}
\frac{2(k+1)N-k(k-1)}{2} \geq \frac{(N-d)(N+d+3)}{2} \Leftrightarrow k+d \geq N.
\end{equation}
The verification of this equivalence is again a straightforward computation.

\smallskip
\noindent {\sc Case 1 :}  $k+d \geq N$.  By (\ref{d-inequality}), we can choose all summands to have value (\ref{value2}). So we must verify that
$$
\frac{(N-d)(N+d+3)}{2} \cdot \frac{-2}{N+2}  \cdot \frac{(k+1)N+2k}{2(k+1)N - k(k-1)} < -(N-d).
$$
This is equivalent to $(k-1)N +4k >0$.  This is clearly satisfied for all $k \geq 1$. Here $k=0$ is not possible since $k+d \geq N$ and $d\leq N-1$.

\smallskip
\noindent {\sc Case 2 :}  $k+d \leq N-1$.  By (\ref{d-inequality}), we can choose at most $\frac{2(k+1)N-k(k-1)}{2}$ summands to have value (\ref{value2}). So we must verify that
\begin{align*}
&\left( \frac{(N-d)(N+d+3)}{2}- \frac{2(k+1)N-k(k-1)}{2}\right) \cdot \frac{-2}{N+1-k} \\
&+ \frac{2(k+1)N-k(k-1)}{2} \cdot
\frac{-2}{N+2}  \cdot \frac{(k+1)N+2k}{2(k+1)N - k(k-1)} < -(N-d).
\end{align*}
A computation shows that this inequality is equivalent to
$$
(N+dN+2d)(N-k-d-1) + 2(N-d) > 0,
$$
which is satisfied since $N-k-d-1 \geq 0$ and $N-d\geq 1$.

\bigskip
(ii) Next, we treat the conditions \lq from infinity\rq\ (\ref{EQ E}), associated to a $Z_j$ of dimension $d\in \{0, \dots, N-1\}$. The only such $Z_j$ with $d=0$ is $\underline{\infty}$; we handle this case first, choosing the value $\frac{-2}{N+3}$ for all the $\frac{N(N+3)}{2}$ summands $\operatorname{Re}(s_{ij})$.  Then, since $\frac{N(N+3)}{2} \cdot \frac{-2}{N+3} = -N$, the condition associated to the blow-up at $\underline{\infty}$ is indeed not satisfied. We verify now that all other conditions are satisfied.

The \lq affine\rq\ conditions, associated to a $Z$ of dimension $k\in \{0,\dots, N-1\}$, are of the form $\binom{N+1-k}{2} \cdot \frac{-2}{N+3} > -(N-k)$, which is equivalent with $k+2>0$.

The other conditions, associated to a $Z$ \lq at infinity\rq\ of dimension $\ell\in \{1,\dots, N-1\}$, are of the form $\frac{(N-\ell)(N+\ell+3)}{2} \cdot \frac{-2}{N+3} < -(N-\ell)$, which is equivalent with $\ell>0$.

\medskip
Finally we treat the cases where $Z_j$ has dimension $d\in \{1, \dots, N-1\}$.
We have that $Z_j$ is contained in $\frac{(N-d)(N+d+3)}{2}$ components of $\bar{A}_N$.
With a similar motivation as in case (i), we put all the $\frac{(N-d)(N+d+3)}{2}$  associated summands $\operatorname{Re}(s_{ij})$ equal to
\begin{equation}\label{value3}
\frac{-2}{N+d+3},
\end{equation}
and all other $\operatorname{Re}(s_{ij})$ equal to
\begin{equation}\label{value4}
\frac{-2}{N+2}\cdot\frac{(d+1)N + 2d}{d(d+3)}.
\end{equation}
Note that in this case, the value (\ref{value3}) is larger than the average and (hence) larger than the value (\ref{value4}).

Clearly, the condition in Theorem \ref{conditions} associated to $Z_j$ is {\em not} satisfied, and one can verify that very similar calculations as in case (i) show that all other conditions in Theorem \ref{conditions} are satisfied.

\section{\label{Section 6}Applications}
In this Section, we describe how our work recovers some results established in \cite{Sussman} and \cite{Brown-Dupont}.

\subsection{Sussman's results on Dotsenko-Fateev integrals}\label{Sussman-Section}
Let 
\begin{equation*}
\Delta _{N}=\left\{ \left( x_{1},\ldots ,x_{N}\right) \in \left[ 0,1\right]
^{N}\mid x_{1}\leq \cdots \leq x_{N}\right\} 
\end{equation*}%
denote the standard $N$-simplex considered as a a subset of $\mathbb{R}^{N}$.   In \cite{Sussman}, Sussman studied Mehta-Selberg integrals of the form
\begin{multline*}
S_{N}\left[ F\right] \left( \alpha ,\beta ,\gamma \right)
=\int\limits_{\Delta _{N}}F\left( x_{1},\ldots ,x_{N}\right)
\prod\limits_{j=1}^{N}x_{j}^{\alpha _{j}}\left( 1-x_{j}\right) ^{\beta _{j}}
\text{ }\times  \\
\prod\limits_{1\leq j<k\leq N}\left( x_{k}-x_{j}\right) ^{2\gamma
_{jk}}\prod\limits_{j=1}^{N}dx_{i},
\end{multline*}
where $F\in C^{\infty }\left( \Delta _{N}\right) $, and $\alpha =\left(
\alpha _{1},\ldots ,\alpha _{N}\right) $, $\beta =\left( \beta _{1},\ldots
,\beta _{N}\right) \in \mathbb{C}^{N}$, $\gamma =\left\{ \gamma _{jk}=\gamma
_{kl}\right\} _{1\leq j<k\leq N}\in \mathbb{C}^{\frac{\left( N-1\right) N}{2}
}$. He established the meromorphic continuation for  these integrals and describes explicitly their polar locus; see  \cite[Theorem 1.1]{Sussman}. 

When we apply our Theorem \ref{maintheorem} to the special case $D=\Delta_N$, we obtain several main results in \cite{Sussman}. More precisely, Theorem \ref{maintheorem} selects the following subset of convergence conditions out of the list of Theorem \ref{conditions}.

\begin{theorem}\label{conditionsDelta}
The zeta function  $ Z_\varphi^{(N)}\left(\Delta_N; \boldsymbol{s}\right)$ converges on the (unbounded) open subset in $\mathbb{C}^{\boldsymbol{d}}$, determined by the following $\frac{N(N+3)}{2}$ inequalities.

\noindent
Only the components $\{x_1=0\}$, $\{x_N=1\}$ and the $N-1$ components $\{x_i=x_{i+1}\}$ (i=1,\dots, N-1) of $A_N$ induce the conditions
\begin{equation}
\operatorname{Re}(s_{01}) > -1 , \quad \operatorname{Re}(s_{N(N+1)}) > -1 \quad\text{ and }\quad  \operatorname{Re}(s_{i(i+1)}) > -1.
\label{EQ A'}
\end{equation}

\noindent
There is only one  exceptional $E_\ell$, coming from a centre $Z_\ell$ of dimension $k\in\{0,\dots,N-2\}$  as in case (a) in (\ref{centres}), inducing a condition, namely $Z_k =\{x_1=\dots=x_{N-k}=0\}$ inducing
\begin{equation}
\sum_{j=1}^{N-k}  \operatorname{Re}(s_{0j})  + \sum_{\substack{1\leq i,j \leq N-k \\ i<j } }   \operatorname{Re}(s_{ij}) > - (N-k).
\label{EQ B'}
\end{equation}

\noindent
There is only one  exceptional $E_\ell$, coming from a centre $Z_\ell$ of dimension $k\in\{0,\dots,N-2\}$ as in case (b) in (\ref{centres}), inducing a condition,  namely $Z_k =\{x_{k+1}=\dots=x_{N}=1\}$ inducing
\begin{equation}
\sum_{j=k+1}^{N}  \operatorname{Re}(s_{j(N+1)})  + \sum_{\substack{k+1\leq i,j \leq N \\ i<j } }   \operatorname{Re}(s_{ij}) > - (N-k).
\label{EQ C'}
\end{equation}

\noindent
There are $k$ exceptional $E_\ell$, coming from  centres $Z_\ell$ of dimension $k\in\{1,\dots,N-2\}$ as in case (c) in (\ref{centres}), inducing a condition, namely $Z=\{x_m=x_{m+1}=\dots=x_{m+N-k}\}$ $ (1\leq m \leq k)$  inducing
\begin{equation}
 \sum_{m\leq i<j \leq m+N-k}    \operatorname{Re}(s_{ij}) > - (N-k).
\label{EQ D'}
\end{equation}
\end{theorem}

\begin{example}\rm
We refer to Example \ref{exampleDelta2} and Figure 4 for  the case $N=2$.  The convergence conditions are
\begin{align*}
&\operatorname{Re}(s_{01}) > -1 ,\quad  \operatorname{Re}(s_{12}) > -1 , \quad \operatorname{Re}(s_{23}) > -1;\\
\operatorname{Re}(s_{01}) +& \operatorname{Re}(s_{02}) + \operatorname{Re}(s_{12}) > -2, \quad
\operatorname{Re}(s_{12}) + \operatorname{Re}(s_{13}) + \operatorname{Re}(s_{23}) > -2,
\end{align*}
which are, of course, a subset of the conditions in Example \ref{conditionsN=2}.
\end{example}

\begin{example}\rm
We sketch the simplex $\Delta_3$ in Figure 6. Compared with Figure 2, we now only put the lines $\{x_1=x_2=0\}$, $\{x_1=x_2=x_3\}$, and $\{x_2=x_3=1\}$ in bold, being precisely the one-dimensional contributing $Z_\ell$.
The nine convergence conditions are
\begin{align*}
\operatorname{Re}(s_{01}) &> -1 ,\quad  \operatorname{Re}(s_{12}) > -1 , \quad \operatorname{Re}(s_{23}) > -1, \quad \operatorname{Re}(s_{34}) > -1 ;\\
\operatorname{Re}(s_{01}) &+ \operatorname{Re}(s_{02}) + \operatorname{Re}(s_{12}) > -2, \quad
\operatorname{Re}(s_{12}) + \operatorname{Re}(s_{13}) + \operatorname{Re}(s_{23}) > -2, \\\
&\operatorname{Re}(s_{23}) + \operatorname{Re}(s_{24}) + \operatorname{Re}(s_{34}) > -2 ;  \\
\operatorname{Re}(s_{01})&+ \operatorname{Re}(s_{02}) + \operatorname{Re}(s_{03}) + \operatorname{Re}(s_{12}) + \operatorname{Re}(s_{13}) + \operatorname{Re}(s_{23})> -3, \\
&\operatorname{Re}(s_{12}) + \operatorname{Re}(s_{13}) + \operatorname{Re}(s_{23})+ \operatorname{Re}(s_{14}) + \operatorname{Re}(s_{24}) + \operatorname{Re}(s_{34}) > -3.
\end{align*}

\end{example}

\centerline{
\beginpicture
\setcoordinatesystem units <.35truecm,.35truecm>

\putrule from 0 8 to 8 8
\setlinear \plot 0 8  5 3 /

\setdashes
\plot 0 0   8 8 /
\setsolid

\setdots <1.5pt>

\putrule from 8 8 to 8 0
\putrule from 0 0 to 8 0
\putrule from -3 -5 to -3 3
\putrule from -3 -5 to 5 -5
\putrule from 5 3 to -3 3
\putrule from 5 3 to 5 -5

\setlinear    \plot  0 8   -3 3 /
\setlinear    \plot  8 0   5 -5 /
\setlinear    \plot  0 0    5 3 /
\setlinear    \plot  0 0   -3 -5 /

\setsolid

\linethickness=1.5pt
\putrule from 0 0 to 0 8

\setplotsymbol ({$\cdot$})

\setlinear    \plot 5 3   8 8 /
\setlinear    \plot  0 0    5 3 /

\put{$\bullet$} at 0 0
\put{$\bullet$} at 5 3
\put{$\underline{0}$} at -0.7 0.7
\put{$\underline{1}$} at 5.7 2.3

\put{Figure 6} at 2 -7

\endpicture
}

\bigskip
The convergence region in Theorem \ref{conditionsDelta} is exactly (formulated with our notations) the one stated in \cite[(8)]{Sussman}.
Using Remark \ref{Gammasum}, our method results also in the expression for $ Z_\varphi^{(N)}\left(\Delta_N; \boldsymbol{s}\right)$ in terms of Gamma functions and the  analytic continuation statement of  \cite[Theorem 1.1 and Corollary 1.1.1]{Sussman}. To be precise, we obtain  \cite[Theorem 1.1 and Corollary 1.1.1]{Sussman} {\em in the sense of distributions}, for then all numbers
$o_{\mathcal{I}}$ in loc. cit. should be taken to be zero.  If we were to develop our results above from the point of view of zeta functions associated to a concrete test function $\varphi$, we would obtain the results in \cite{Sussman} as stated there.

\bigskip

Sussman also studied a variation of $Z_\varphi^{(N)}\left( \square_N; \boldsymbol{s}\right) $, namely
\begin{equation}
I_\varphi^{(N)}\left( \square_N; \boldsymbol{s}\right)  :=
{\displaystyle\int\limits_{\square_N}} \varphi(x)
{\displaystyle\prod\limits_{i=1}^{N}}
\left\vert x_{i}\right\vert^{s_{0i}}
{\displaystyle\prod\limits_{i=1}^{N}}\left\vert 1-x_{i}\right\vert^{s_{i(N+1)}}\text{ }
{\displaystyle\prod\limits_{1\leq i<j\leq  N}}
\left\vert x_{i}-x_{j} +i0 \right\vert^{s_{ij}}
dx. \label{zeta_function_string_variation}
\end{equation}

\medskip
First, for $Z_\varphi^{(N)}\left( \square_N; \boldsymbol{s}\right) $ itself, Theorem \ref{maintheorem} selects as convergence conditions precisely all the conditions (\ref{EQ A}),
 (\ref{EQ B}), (\ref{EQ C}) and (\ref{EQ D}) from Theorem \ref{conditions}.  

\medskip
We refer to \cite[Vol. 1, 3.6]{G-S} for the definition of the (generalized) function $(x+i0)^\lambda$. Crucial here is that this is an entire function in $\lambda$. As a consequence, each $Z_\ell$ of the form  $\{x_i-x_j=0\}$ or an intersection of these will not even induce a candidate polar hyperplane for $I_\varphi^{(N)}\left( \square_N; \boldsymbol{s}\right)$.

Concretely, our Theorem \ref{maintheorem} then provides the same convergence conditions for $I_\varphi^{(N)}\left( \square_N; \boldsymbol{s}\right)$  as given by \cite[Theorem 1.3]{Sussman},
namely the list above without the ones \lq coming from diagonals\rq. More precisely, these conditions are $\operatorname{Re}(s_{0i}) > -1$ and $\operatorname{Re}(s_{i(N+1)}) > -1$ for $i=1,\dots,N$, and all the conditions  (\ref{EQ B}) and  (\ref{EQ C}) from Theorem \ref{conditions}.

 Moreover, because of our Proposition \ref{independence}, we can assert that this statement is sharp: having an analytic continuation to a larger open set is impossible.

\subsection{General $N$-point genus zero open string amplitudes}

We follow the notation from Section 3.2 in the paper \cite{Brown-Dupont} of Brown and Dupont. The $N$-point
genus zero open string amplitude, $N=n+3$, is formally defined as
\begin{equation*}
I^{\text{open}}\left( \omega ,\boldsymbol{s}\right)
=\int\limits_{0<t_{1}<\cdots <t_{n}<1}\text{ \ }\prod\limits_{0\leq i<j\leq
n+1}\left( t_{j}-t_{i}\right) ^{s_{i}{}_{j}}\omega ,
\end{equation*}
where $s_{ij}\in \mathbb{C}$,  with the convention $t_{0}=0$, $t_{n+1}=1$, and with $\omega $  a  differential form  of type
\begin{equation*}
\omega =\frac{dt_{1}\wedge \cdots \wedge dt_{n}}{\prod\limits_{i=0}^{n}
\left( t_{\sigma \left( i+1\right) }-t_{\sigma \left( i\right) }\right) },
\end{equation*}
where $\sigma $ is a permutation of $\left\{ 0,1,\ldots ,n+1\right\} $.
So $I^{\text{open}}\left( \omega ,\boldsymbol{s}\right)$ is, 
up to a translation in some variables, of the form 
$\pm Z_{\varphi
}^{\left( n\right) }\left( \Delta_n;\boldsymbol{s}\right) $, where $\varphi \equiv 1$, and hence can be considered as 
 a particular case of our integrals $Z_{\varphi }^{\left( N\right) }\left( D;\boldsymbol{s}\right) $. Our results
provide a region of convergence and a meromorphic continuation for $I^{\text{open}}\left( \omega ,\boldsymbol{s}\right) $ in the parameters $s_{ij}\in 
\mathbb{C}$, which is an alternative proof for the meromorphic regularization of $I^{\text{open}}\left( \omega ,\boldsymbol{s}\right) $  in \cite[Section 4]{Brown-Dupont}.
Furthermore, we provide an explicit description of the
polar locus as a finite union of hyperplanes, whose equations can be
computed.
If $ Z_{\varphi
}^{\left( n\right) }\left( \Delta_n;\boldsymbol{s}\right)$ is interpreted as a string amplitude, then its poles correspond to the mass spectrum of the string theory.

On the other hand,
 in \cite{Brown-Dupont}, the authors provide many other results, for instance, an interpretation of the poles in the Mandelstam variables $s_{ij}$ in terms of the poles of $\omega$, as well as Laurent expansions with multiple zeta values as coefficients. Note that the sums in Theorem 1.1 and Theorem 4.20 in \cite{Brown-Dupont} are \lq in spirit\rq\ related to sums over various charts in an embedded resolution (with respect to some partition of unity). 
In both settings, the integral is computed as a finite sum of more \lq accessible\rq\ integrals, which can be proven to converge on some domain, and which are often reduced to lower dimensional integrals.
 

\section{Extension of the main results to general hyperplane arrangements}\label{generalizations}

An essential extra feature of our conceptual proofs of Propositions \ref{if} and \ref{only if} is that they also apply to similar zeta functions associated with {\em arbitrary} hyperplane arrangements.

Let $L_i(x) \in \mathbb{R}[x_1,\dots,x_N]$ be a polynomial of degree $1$, for $i=1,\dots,d$. Put $f(x):=\prod_{i=1}^d L_i(x)$ and let $A:= f^{-1}(0)$ be the induced affine hyperplane arrangement in $\mathbb{R}^N$ with irreducible components $Z_i:=L_i^{-1}(0)$.
 We consider the associated multivariate zeta function
\begin{equation}
Z_\varphi\left( f, D; \boldsymbol{s}\right)  :=
{\displaystyle\int\limits_D} \varphi(x)
{\displaystyle\prod\limits_{i=1}^{d}}
\left\vert L_{i}(x)\right\vert^{s_i}
dx, \label{zeta_function_string_general}
\end{equation}
where the function $\varphi:\mathbb{R}^N\to \mathbb{C}$ is  smooth on the (closed) integration domain $D$.
Here we allow as integration domain $D$  \emph{all} possible polyhedra with boundary conditions given by inequalities coming from components of $A$, that is, $D$ is given by conditions of the form $L_i(x) \geq 0$  or $L_i(x) \leq 0$ for some indices $i$ such that $\operatorname{dim}(D)=N$ (including the case \lq no conditions\rq, being $D=\mathbb{R}^N$). We also assume that   $Z_\varphi\left( D; \boldsymbol{s}\right)$ converges for all $\boldsymbol{s}$ in a nonempty open part of $\mathbb{C}^d$; this is automatic when $\varphi$ has compact support or when $D$ is bounded.

\medskip

We recall some standard terminology in the theory of hyperplane arrangements in $\mathbb{R}^N$. 
An arrangement $A$ is {\em central} if 
$\cap_{\ell=1}^d Z_\ell \neq \emptyset$.
A central hyperplane arrangement $A$ is {\em indecomposable} if there is no linear change of coordinates on $\mathbb{R}^N$
such that $f$ can be written as the product of two non-constant polynomials in disjoint sets of variables.

 An {\em edge} of $A$ is any (nonempty) intersection of components of $A$. 
An edge $W$  of $A$
is called {\em dense} if the central arrangement $A_W:= \{Z_\ell \text{ in } A \mid Z_\ell \supset W\}$ is indecomposable.
(In particular, every component $Z_\ell $ itself is a dense edge of $A$.)
One can compute an embedded resolution of $A$ as follows, by blowing up only along dense edges, see  \cite[Theorem 3.1]{STV}, also \cite[10.8]{Va}.

\begin{theorem}\label{dense edge resolution}
 Let $\pi_0: X_1 \to X_0=\mathbb{R}^N$ be the blow-up of all zero-dimensional dense edges of $A$. In general, let $\pi_k: X_{k+1} \to X_{k}$ be the blow-up of all (strict transforms of) $k$-dimensional dense edges of $A$, for  $k =  0,\dots, N-2$,  in some order.
 Then $\pi=\pi_0 \circ \dots \circ \pi_{N-2}$ is an embedded resolution of $A\subset \mathbb{R}^N$.
\end{theorem}

As in Subsection \ref{resolution}, we consider the compactification $(\mathbb{P}_{\mathbb{R}}^1)^N$ of $\mathbb{R}^N$, and the induced hyperplane arrangement $\bar{A}$ in $(\mathbb{P}_{\mathbb{R}}^1)^N$, consisting of the closures of the components of $A$ in $(\mathbb{P}_{\mathbb{R}}^1)^N$, together with the $N$ \lq hyperplanes at infinity\rq. 
The notions above, as well as Theorem \ref{dense edge resolution}, have an immediate extension to $\bar{A}$.

One can verify that our \lq economic\rq\  embedded resolution $\pi$ of the hyperplane arrangement $A_N$ or $\bar{A}_N$, constructed in Subsection \ref{resolution}, is precisely this composition of blow-ups along only dense edges!

\bigskip  The meromorphic continuation of $Z_\varphi\left(f, D; \boldsymbol{s}\right)$ follows the general theory, as recalled in Subsection \ref{zetafunctions}. 
Each dense edge $Z_\ell$  induces a convergence condition, being a linear inequality $\mathcal{L}_\ell $ in the $\operatorname{Re}(s_i)$ of the form $\sum_{i=1}^m a_i^{(\ell)} s_i +b^{(\ell)} >0$,
where the $a_{i}^{(\ell)}$ are integers (not all zero) and $b^{(\ell)}$ is an integer, 

For an affine $Z_\ell$ of dimension $k$, this inequality has the easy form 
$$\sum_{\substack{Z_i\supset Z_\ell \\\dim(Z_i)=N-1 } }  \operatorname{Re}(s_i)  > -(N-k).$$
So one has the following result.

\begin{theorem}\label{thm general hyperplanes}
Assume that the zeta function   $Z_\varphi\left( D; \boldsymbol{s}\right)$ converges for all $\boldsymbol{s}$ in some nonempty open subset of $\mathbb{C}^d$ (this is automatic when $\varphi$ has compact support or when $D$ is bounded).
Then $Z_\varphi\left(f, D; \boldsymbol{s}\right)$ is convergent and holomorphic in the open domain, determined by the inequalities $\mathcal{L}_\ell$, associated to the dense edges $Z_\ell$ of $A$, i.e., in the domain
$$
\left\{  (s_{1},\ldots,s_{d})\in\mathbb{C}^{d} \mid \sum_{i=1}^{m}a_{i}^{(\ell)}
\operatorname{Re}(s_{i})+b^{(\ell)}>0\right\}  .
$$
Furthermore, it admits an analytic continuation to the whole $\mathbb{C}^{d}$, as a meromorphic function with polar locus contained in
\[
\bigcup_{t\in\mathbb{N}}\left\{  \sum_{i=1}^{m}a_{i}s_{i}^{(\ell)}+b^{(\ell)}+t=0\right\}.
\]
\end{theorem}

Now the same proof as for Proposition \ref{only if} yields the following analogous statement for arbitrary arrangements!

\begin{proposition}\label{prop general hyperplanes}
Let $Z_j$ be a dense edge of $A$ or $\bar{A}$.
If $\operatorname{dim}(Z_j\cap D) < \operatorname{dim} (Z_j)$, then $Z_j$ does not contribute to the polar locus of $Z_\varphi \left(f,D;  \boldsymbol{s}\right)$.
\end{proposition}

Also, the core of Proposition \ref{if} is more generally true for any hyperplane arrangement. That is, the same proof is valid till before the conclusion using Proposition \ref{independence}. 
Concretely: {\em if $\operatorname{dim}(Z_j \cap D) = \operatorname{dim} (D)$, then $E_j$ intersects $\tilde{D}$ in its interior or in one of its facets, and this intersection has dimension $N-1$}.  So the condition $\mathcal{L}_j$ is needed to assure convergence on some part of the embedded resolution space.  It could be possible, however, that all these conditions together are not independent.  

\begin{example}
The Mehta-Macdonald integrals, see Subsection \ref{Mehta-Macdonald Integrals} below, are specializations of (\ref{zeta_function_string_general}).
We can apply the above theorem to the following generalization of these integrals. We denote by $\mathcal{S}\left( \mathbb{R}^{N}\right) $ the
Schwartz space, which is the space of smooth functions from $\mathbb{R}^{N}$
 into $\mathbb{C}$, whose derivatives are rapidly decreasing. For $\beta \in 
\mathbb{C}$, with $\operatorname{Re}\left( \beta \right) >0$, and $\varphi $ a
Schwartz function\textit{, }we define
\begin{equation*}
Z_{\varphi }^{(N)}\left( \mathbb{R}^{N},\beta \right) =\int\limits_{\mathbb{R}^{N}}\varphi \left( x\right) \prod\limits_{i=1}^{d}\left\vert
L_{i}(x)\right\vert ^{2\beta }dx.
\end{equation*}
Then, by a well-known argument, the integral $Z_{\varphi }^{(N)}\left( 
\mathbb{R}^{N},\beta \right) $ defines a holomorphic function in the
half-plane $\operatorname{Re}\left( \beta \right) >0$. The integral  $Z_{\varphi }^{(N)}\left( \mathbb{R}^{N},\beta \right) $ has a meromorphic continuation to the whole $\mathbb{C}$, see  \cite{Igusa}, \cite{Igusa-old}; the description of its polar locus is a consequence of Theorem \ref{thm general hyperplanes}.
\end{example}

\section{\label{Section 8} Specific integrals}

The  Selberg-Mehta-Macdonald and Dotsenko-Fateev-like integrals appear in several areas of mathematics and physics; the literature about these integrals is vast, so our list of bibliographic references is far from complete. Below we review some of them as motivation for the integrals studied here.

\subsection{Mehta-Selberg Integrals}

The Selberg integral
\begin{eqnarray*}
S_{N}\left( \alpha ,\beta ,\gamma \right)  &:&=\int\nolimits_{0}^{1}\cdots
\int\nolimits_{0}^{1}\prod\limits_{i=1}^{N}t_{i}^{\alpha-1}\left(1-t_{i}\right) ^{\beta -1}\prod\limits_{1\leq i<j\leq N}\left\vert
t_{i}-t_{j}\right\vert ^{2\gamma }\prod\limits_{i=1}^{N}dt_{i} \\
&=&\prod\limits_{j=0}^{N-1}\frac{\Gamma \left( \alpha +j\gamma \right)
\Gamma \left( \beta +j\gamma \right) \Gamma \left( 1+\left( j+1\right)
\gamma \right) }{\Gamma \left( \alpha +\beta +\left( N+j-1\right) \gamma
\right) \Gamma \left( 1+\gamma \right) },
\end{eqnarray*}
for 
\begin{equation*}
\operatorname{Re}\left( \alpha \right) >0\text{, }\operatorname{Re}\left( \beta \right) >0%
\text{, }\operatorname{Re}(\gamma )>-\min \left\{ \frac{1}{N},\frac{\operatorname{Re}\left(
\alpha \right) }{N-1},\frac{\operatorname{Re}\left( \beta \right) }{N-1}\right\} ,
\end{equation*}
where $\Gamma \left( s\right) $ is the Gamma function; see \cite{Selberg}, \cite{Forrester et al}.

The Mehta integral 
\begin{eqnarray*}
F_{N}(\gamma ) &:&=\frac{1}{\left( 2\pi \right) ^{\frac{N}{2}}}
\int\nolimits_{-\infty }^{\infty }\cdots \int\nolimits_{-\infty }^{\infty
}\prod\limits_{i=1}^{N}e^{-\frac{t_{i}^{2}}{2}}\prod\limits_{1\leq i<j\leq
N}\left\vert t_{i}-t_{j}\right\vert ^{2\gamma }\prod\limits_{i=1}^{N}dt_{i}
\\
&=&\prod\limits_{j=1}^{n}\frac{\Gamma \left( 1+j\gamma \right) }{\Gamma
\left( 1+\gamma \right) },
\end{eqnarray*}
for $\operatorname{Re}(\gamma )>\frac{-1}{N}$. The integral $\left( 2\pi \right) ^{
\frac{N}{2}}F_{N}(\frac{\beta }{2})$ is the partition function associated with the probability measure
\begin{equation}
\frac{e^{-\beta H}}{\left( 2\pi \right) ^{\frac{N}{2}}F_{N}(\frac{\beta }{2})
}\prod\limits_{i=1}^{N}dt_{i},  \label{Coullomb_measure}
\end{equation}
where 
\begin{equation}
H=\frac{1}{2\beta }\sum\limits_{i=1}^{N}t_{i}^{2}-\sum\limits_{1\leq
i<j\leq N}\log \left\vert t_{i}-t_{j}\right\vert .  \label{Potential}
\end{equation}
In fact, (\ref{Coullomb_measure})-(\ref{Potential}) describe a gas of $n$
particles on the line, at inverse temperature $\beta $, interacting through the
repulsive Coulomb potential and confined by a harmonic well. The Selberg integral can be used to evaluate Mehta's integral; see \cite{Forrester et al},
\cite{Forrester}, \cite{Mehta}, and the references therein.

\subsection{Mehta-Macdonald Integrals}\label{Mehta-Macdonald Integrals}

The Macdonald conjectures involve a generalization of Mehta
integral, \cite{Macdonald}, \cite{Forrester et al}. In 1982, Macdonald
published  several conjectures generalizing the Mehta integral. Let $G$
be a finite group of isometries of $\mathbb{R}^{N}$, generated by
reflections in $d$ hyperplanes. Take the equations for the hyperplanes of
the form%
\begin{equation*}
L_{i}(x)=a_{i,1}x_{1}+\ldots +a_{i,N}x_{N},\text{ with }a_{i,1}^{2}+\ldots
+a_{i,N}^{2}=2,
\end{equation*}
for $i=1,\ldots ,d$, and set
\begin{equation*}
f(x):=\prod\limits_{i=1}^{d}L_{i}(x).
\end{equation*}
By its action on $\mathbb{R}^{N}$, the group $G$ acts on polynomials in $
x=(x_{1},...,x_{N})$. The polynomials that are invariant under the action of 
$G$ are referred to as $G$-invariant polynomials. They form an $\mathbb{R}$
-algebra $\mathbb{R}[g_{1},...,g_{N}]$ generated by d algebraically
independent polynomials of degrees $e_{1},\ldots ,e_{N}$. These 
polynomials are uniquely determined by the underlying reflection group. With
this notation, the Macdonald conjecture for a finite reflection group $G$
asserts
\begin{equation*}
\frac{1}{\left( 2\pi \right) ^{\frac{N}{2}}}\int\limits_{\mathbb{R}^{N}}e^{-\frac{\left\vert x\right\vert ^{2}}{2}}\left\vert f(x)\right\vert ^{2\beta
}dx=\prod\limits_{i=1}^{N}\frac{\Gamma \left( 1+e_{i}\beta \right) }{\Gamma
\left( 1+\beta \right) }.
\end{equation*}
For a further discussion the reader may consult \cite{Forrester et al}, and
the references therein.

\subsection{Dotsenko-Fateev integrals}

Another important generalization of Mehta integrals are Dotsenko-Fateev integrals  that appear in conformal quantum field theory, \cite{Dotsenko-Fateev},
\begin{eqnarray*}
&&\operatorname{PV}\int\limits_{\left[ 0,1\right] ^{p}}\text{ }\int\limits_{\left[
1,\infty \right) ^{N-p}}\text{ }\int\limits_{\left[ 0,1\right] ^{r}}\text{ }
\int\limits_{\left[ 1,\infty \right) ^{m-r}}\text{ }\prod
\limits_{i=1}^{N}t_{i}^{\alpha }\left( 1-t_{i}\right) ^{\beta
}\prod\limits_{i=1}^{m}\tau _{i}^{\alpha ^{\prime }}\left( \tau
_{i}-1\right) ^{\beta ^{\prime }}\times  \\
&&\prod\limits_{i=1}^{N}\prod\limits_{j=1}^{m}\left( \tau
_{j}-t_{i}\right) ^{-2}\prod\limits_{1\leq i<j\leq N}\left\vert
t_{i}-t_{j}\right\vert ^{2\gamma }\prod\limits_{1\leq i<j\leq m}\left\vert
\tau _{i}-\tau _{j}\right\vert ^{2\gamma ^{\prime
}}\prod\limits_{i=1}^{N}dt_{i}\prod\limits_{i=1}^{m}d\tau _{i},
\end{eqnarray*}
where PV denotes the principal value, 
\begin{equation*}
\frac{\alpha }{\alpha ^{\prime }}=\frac{\beta }{\beta ^{\prime }}=-\gamma 
\text{, }\gamma \gamma ^{\prime }=1\text{, }0\leq p\leq N\text{, }0\leq
r\leq m.
\end{equation*}
In the case $p=N$ and $m=0$, the Dotsenko-Fateev integral is, up to a shift
by $1$ in $\alpha $\ and $\beta $, precisely the Selberg integral. Dotsenko
and Fateev also consider a complex generalization of the Selberg integral,
which was also studied independently by Aomoto in \cite{Aomoto}, namely
\begin{equation*}
A_{N}(\alpha ,\beta ,\gamma )=\int\limits_{\mathbb{R}^{2}}\cdots
\int\limits_{\mathbb{R}^{2}}\text{ }\prod\limits_{i=1}^{N}\left\vert 
\boldsymbol{r}_{i}\right\vert ^{2\left( \alpha -1\right) }\left\vert 
\boldsymbol{u}-\boldsymbol{r}_{i}\right\vert ^{2\left( \beta -1\right)
}\prod\limits_{1\leq i<j\leq N}\left\vert \boldsymbol{r}_{i}-\boldsymbol{r}
_{j}\right\vert ^{4\gamma }\prod\limits_{i=1}^{N}d\boldsymbol{r}_{i},
\end{equation*} where $\boldsymbol{u}\in \mathbb{R}^{2}$ is an arbitrary unit vector. Dotsenko and Fateev, as well as  Aomoto showed that 
\begin{equation*}
A_{N}(\alpha ,\beta ,\gamma )=\frac{S_{N}^{2}\left( \alpha ,\beta ,\gamma
\right) }{N!}\prod\limits_{j=0}^{N}\frac{\sin \pi \left( \alpha +j\gamma
\right) \sin \pi \left( \beta +j\gamma \right) \sin \pi \left( 1+j\right)
\gamma }{\sin \pi \left( \alpha +\beta +\left( N+j-1\right) \gamma \right)
\sin \pi \gamma },
\end{equation*}%
for 
\begin{equation*}
\operatorname{Re}\left( \alpha +\beta +\left( N-1\right) \gamma \right) <1\text{ and 
}\operatorname{Re}(\alpha +\beta +2\left( N-1\right) \gamma )<1.
\end{equation*}

\subsection{Local zeta functions for graphs and Log-Coulomb Gases}

In \cite{Zuniga-JMP-2022}, a generalization of the Mehta integral of the form%
\begin{equation*}
Z_{\varphi }(\boldsymbol{s})=\int\limits_{-\infty }^{\infty }\cdots
\int\limits_{-\infty }^{\infty }\varphi \left( x_{1},\ldots ,x_{N}\right)
\prod\limits_{1\leq i<j\leq N}\left\vert x_{i}-x_{j}\right\vert
^{s_{ij}}\prod\limits_{i=1}^{N}dx_{i},
\end{equation*}
where $\varphi $ is a Schwartz function, and $\boldsymbol{s=}\left(
s_{ij}\right) _{1\leq i<j\leq N}\in \mathbb{C}^{\frac{N(N-1)}{2}}$ with $%
\operatorname{Re}\left( s_{ij}\right) >0$ for any $1\leq i<j\leq N$, were studied.
The original Mehta integral $F_{N}(\gamma )$ is exactly $F_{N}(\gamma )=%
\frac{1}{\left( 2\pi \right) ^{\frac{N}{2}}}Z_{\varphi }(\boldsymbol{s})\mid
_{s_{ij}=2\gamma }$, with $\varphi \left( x_{1},\ldots ,x_{N}\right) =$ $e^{-%
\frac{1}{2}\sum_{i=1}^{N}x_{i}^{2}}$. The integral $Z_{\varphi }(\boldsymbol{%
s})$ is a particular case of a multivariate local zeta function. These
functions admit meromorphic continuations to the whole $\mathbb{C}^{\frac{%
N(N-1)}{2}}$, see e.g. \cite{Loeser}. Today, there exists a uniform
theory of local zeta functions in local fields of characteristic zero,
e.g. $\left( \mathbb{R},\left\vert \cdot \right\vert \right) $, $\left( 
\mathbb{C},\left\vert \cdot \right\vert \right) $, and the field of $p$-adic
numbers $\left( \mathbb{Q}_{p},\left\vert \cdot \right\vert _{p}\right) $,
see \cite{Igusa-old}, \cite{Igusa}, see also \cite{Denef}, \cite{G-S}, \cite%
{Loeser}, \cite{Veys-Zuniga-Advances}, \cite{Zuniga-Veys2} and the references therein. 

Given a local field $(\mathbb{K},\left\vert \cdot \right\vert _{\mathbb{K}})$%
, for instance $\mathbb{R}$, $\mathbb{C}$, $\mathbb{Q}_{p}$, and a finite,
simple graph $G$, we attach to them a $1D$ log-Coulomb gas and a local zeta
function. By a gas configuration\ we\ mean a triple $\left( \boldsymbol{x},%
\boldsymbol{e},G\right) $, with $\boldsymbol{x}=\left( x_{v}\right) _{v\in
V(G)}$, $\boldsymbol{e}=\left( e_{v}\right) _{v\in V(G)}$, where $e_{v}\in 
\mathbb{R}$ is a charge located at the site $x_{v}\in \mathbb{K}$, and the interaction between the charges is determined by the graph $G$. Given a
vertex $u$ of $G$ ($u\in V(G)$), the charged particle at the site $x_{u}$ can interact only with the particles located at the sites $x_{v}$ for which there exists an edge between $u$ and $v$ (we denote this fact as $u\sim v$).
The Hamiltonian is given by 
\begin{equation}
H_{\mathbb{K}}(\boldsymbol{x};\mathbf{e},\beta ,\Phi ,G)=-\sum\limits
_{\substack{ u,v\in V\left( G\right)  \\ u\sim v}}\ln \left\vert
x_{u}-x_{v}\right\vert _{\mathbb{K}}^{e_{u}e_{v}}+\frac{1}{\beta }P(%
\boldsymbol{x}),  \label{Hamiltonian}
\end{equation}%
where $\beta =\frac{1}{k_{B}T}$ (with $k_{B}$ the Boltzmann constant, $T$
the absolute temperature), $P:\mathbb{K}^{\left\vert V(G)\right\vert }%
\mathbb{\rightarrow R}$ is a confining potential such that $\Phi \left( 
\boldsymbol{x}\right) =e^{-P(\boldsymbol{x})}$ is a test function, which
means that $P=+\infty $ outside of a compact subset.

The partition function attached to the Hamiltonian (\ref{Hamiltonian}) is
given by%
\begin{equation}
\mathcal{Z}_{G,\mathbb{K},\Phi,\mathbf{e}}\left( \beta\right) =\int
\limits_{\mathbb{K}^{\left\vert V(G)\right\vert }}\text{ }\Phi\left( 
\boldsymbol{x}\right) \prod \limits_{_{\substack{ u,v\in V\left( G\right) 
\\ u\sim v}}}\ \left\vert x_{u}-x_{v}\right\vert _{\mathbb{K}%
}^{e_{u}e_{v}\beta}\prod \limits_{v\in V\left( G\right) }dx_{v}.
\label{Partition_function}
\end{equation}
In order to study this integral, using geometric techniques, it is convenient to extend $e_{u}e_{v}\beta$ to a complex variable $s\left(
u,v\right) $, in this way the partition function (\ref{Partition_function})
becomes a local zeta function. Then the partition function is recovered from
the local zeta function taking $s\left( u,v\right) =e_{u}e_{v}\beta$.

The local zeta function attached to $G$, $\Phi $ is defined as 
\begin{equation*}
Z_{\Phi }(\boldsymbol{s};G,\mathbb{K})=\int\limits_{\mathbb{K}^{\left\vert
V(G)\right\vert }}\text{ }\Phi \left( \boldsymbol{x}\right) \prod\limits 
_{\substack{ u,v\in V(G)  \\ u\sim v}}\left\vert x_{u}-x_{v}\right\vert _{%
\mathbb{K}}^{s\left( u,v\right) }\prod\limits_{v\in V(G)}dx_{v}\text{,}
\end{equation*}%
where $\boldsymbol{s}=\left( s\left( u,v\right) \right) $ for $u,v\in V(G)$
for $u\sim v$, $s\left( u,v\right) $ is a complex variable attached to the edge connecting the vertices $u$ and $v$, and $\prod\nolimits_{v\in V(G)}dx_{v}$ is a Haar measure of the locally compact group $(\mathbb{K}%
^{\left\vert V(G)\right\vert },+)$. The integral converges for $\operatorname{Re}
(s\left( u,v\right) )>0$ for any $\left( u,v\right) $.\ The partition
function $\mathcal{Z}_{G,\mathbb{K},\Phi ,\mathbf{e}}\left( \beta \right) $
of $H_{\mathbb{K}}(\boldsymbol{x};\mathbf{e},\beta ,\Phi ,G)$ is related to the local zeta function of the graph by%
\begin{equation*}
\mathcal{Z}_{G,\mathbb{K},\Phi ,\mathbf{e}}\left( \beta \right) =\left.
Z_{\Phi }(\boldsymbol{s};G,\mathbb{K})\right\vert _{s\left( u,v\right)
=e_{u}e_{v}\beta }.
\end{equation*}%
The zeta function $Z_{\Phi }(\boldsymbol{s};G)$ admits a meromorphic
continuation to the whole complex space $\mathbb{C}^{\left\vert
E(G)\right\vert }$, see \cite[Th\'{e}or\`{e}me 1.1.4]{Loeser}.

\bigskip
\noindent
{\sc Acknowledgement.}
We want to thank the referees for carefully reading the manuscript and for various suggestions improving the presentation of the paper.

\end{document}